\documentclass[12pt,draftcls,onecolumn]{IEEEtran}

\usepackage{amsmath}
\usepackage{amssymb}
\usepackage{amsthm}
\usepackage{graphicx}

\graphicspath{Figures/}

\newtheorem{theorem}{Theorem}
\newtheorem{definition}{Definition}
\newtheorem{lemma}{Lemma}
\newtheorem{corollary}{Corollary}
\ifCLASSINFOpdf
\else
\fi
\begin{document}
%
\title{A Supermodular Optimization Framework for Leader Selection under Link Noise in Linear Multi-Agent Systems}
%
%
%

\author{Andrew~Clark,
        Linda~Bushnell\footnote{$^{1}$Corresponding author.  Email: lb2@uw.edu}$^{1}$,
        and~Radha~Poovendran \\
        Department of Electrical Engineering\\
         University of Washington, Seattle, WA 98195\\
        \{awclark, lb2, rp3\}@uw.edu}

\maketitle

\begin{abstract}
\noindent 
In many applications of multi-agent systems (MAS),  a set of leader agents acts as a control input to the remaining follower agents.  In this paper, we introduce an analytical approach to selecting leader agents  in order to minimize the  total mean-square error of the follower agent states from their desired value in steady-state in the presence of noisy communication links.  We show
 that the problem of choosing leaders in order to minimize this error can be solved using supermodular optimization techniques, leading to efficient algorithms that are within a provable bound of the optimum.  We formulate two leader selection problems within our framework, namely the problem of choosing a fixed number of leaders to minimize the error, as well as the problem of choosing the minimum number of leaders to achieve a tolerated level of error.  We study both  leader selection criteria for different scenarios, including MAS with  static topologies, topologies experiencing random link or node failures, switching topologies, and topologies that vary arbitrarily in time due to node mobility.  In addition to providing provable bounds for all these cases, simulation results demonstrate that  our approach outperforms other leader selection methods, such as node degree-based and random selection methods, and provides comparable performance to current state of the art algorithms.
\end{abstract}

\section{Introduction}
\label{sec:intro}

Multi-agent systems (MAS) consist of networked, autonomous agents, where
each agent receives inputs from its neighbors, uses the inputs to perform computations and update its state information, and broadcasts the resulting information as output to its neighbors.
The MAS framework has been used to model and analyze man-made systems in a wide variety of applications, including the power grid~\cite{MAS_power_grid}, formations of unmanned vehicles~\cite{olfati2006flocking}, sensor networks~\cite{olfati2005consensus}, and networks of nanoscale devices~\cite{MAS_nanosystem}.
Natural phenomena, such as flocking in birds~\cite{reynolds1987flocks}, have also been  modeled as MAS.

An important sub-class of MAS consists of leader-follower systems, in which a set of leader agents act as control inputs to the remaining agents~\cite{ren2005coordination}.
Each follower agent computes its state value based on the states of its neighbors, which may include one or more leader agents.  Hence,  controlling the states of the leaders  influences the dynamics of the follower agents.
A leader-follower system can therefore be viewed as a controlled system, in which the follower agents act as the plant and the leader agents act as control inputs~\cite{ji2006leader}.

Existing control-theoretic work on leader-follower systems has shown that the  performance of the system, including the level of error in the follower node states due to link noise~\cite{wang2009robustness,patterson2010leader,lin2011optimal}, depends on which agents act as leaders.  Errors in a follower agent's state occur when the inputs  from its neighbors are corrupted by noise, causing the agent to update its state based on incorrect information.
The affected agent will  broadcast its updated state information to its neighbors, which update their states based on the received information, causing state errors to propagate through the MAS.  The choice of leader agents determines the level of error in the follower agent states due to the propagation of leader inputs through noisy communication links.  

In spite of the effect of the choice of leader agents on MAS performance, the design of algorithms for selecting agents to act as leaders is currently in its early stages.
Since the number of possible leader sets grows exponentially in the number of leaders and the total number of agents, an exhaustive search over all leader sets is impractical.
An analytical approach for choosing leaders in order to ensure that the MAS is controllable from its leader agents was introduced in~\cite{liu2011controllability}.   This approach, however, does not consider the impact of noise in the communication links between agents, leading to deviations from the desired behavior when even a single link experiences noise.

Leader selection algorithms based on convex optimization have been proposed for static networks in order to minimize the error in the follower agent states due to noise in~\cite{fardad2011noisefree} and \cite{lin2011noisecorrupted}.  These convex optimization-based algorithms, however, do not provide provable guarantees on the optimality of the resulting leader set.  Currently there is no analytical framework for leader selection for minimizing error due to noise that provides such guarantees.


In this paper, we present an analytical approach to solve the problem of selecting the leader set that minimizes  the overall system error, defined as the mean-square error of  the follower agent states from their desired steady-state value.  We formulate the problem of selecting the optimal set of leaders as a set optimization problem, and present a solution framework based on supermodular optimization.  Our framework leads to efficient algorithms that provide provable bounds on the gap between the mean-square error resulting from the  leader set under our framework and the minimum possible error in both static and dynamic networks.  
We make the following specific contributions:
\begin{itemize}
\item We develop a supermodular optimization framework for choosing leaders in a linear MAS in order to minimize the sum of the mean-square errors of the follower agent states.
\item We prove that the mean-square error due to link noise is a supermodular function of the set of leader agents by  observing that the error of each follower agent's state is proportional to the commute time of a random walk between the follower agent and the leader set\footnote[2]{The formulation in this paper differs from that of~\cite{clark2011submodular}, which selected leaders based on minimizing the effective resistance of an equivalent electrical network.}.  We  then show that the commute time is a supermodular function of the leader set.
\item We analyze two classes of the leader selection problem within our framework:  the problem of choosing a fixed number of leaders in order to minimize the error due to noise, and  the problem of finding the minimum number of leaders needed, as well as the identities of the leaders, in order to meet a given error bound.
\item We extend our approach to a broad class of multi-agent system topologies, including systems with: (1)  static network topology, (2)  random link and node failures, (3) switching between predefined topologies, and (4)  network topologies that vary arbitrarily in time.
\item 
    We compare our results with other leader selection methods, including random heuristics and choosing high- and average-degree agents, through simulation and show that our supermodular approach outperforms both schemes.  We also show that the supermodular optimization approach provides comparable performance to state-of-the-art methods based on convex optimization.  
\end{itemize}

The  paper is organized as follows.  In Section \ref{sec:related}, we review  related work on leader-follower MAS.  Section \ref{sec:preliminaries} states our basic definitions and assumptions and gives background on supermodular functions.  In Section \ref{sec:problem_static}, we formulate the leader selection problem for the case of MAS with a static network topology and derive a supermodular optimization solution.  In Section \ref{sec:problem_dynamic}, we analyze leader selection in networks with time-varying, dynamic topologies.  Section \ref{sec:simulation} evaluates our approach and compares with other widely-used leader selection algorithms through a simulation study.  Section \ref{sec:conclusion} presents our conclusions. 
\section{Related Work}
\label{sec:related}


The impact of a given leader set on system performance was considered in \cite{tanner2004controllability}, where the states of the follower agents are treated as the plant to be controlled, while the leader states act as control inputs.  In the case of  linear MAS, the plant dynamics are given by the graph Laplacian of the subgraph defined by the follower agents.  In \cite{tanner2004controllability}, it was shown that the leader-follower system is controllable from the leader agents if and only if the Laplacian eigenvalues are distinct.  An alternative condition for controllability, based on the automorphism group of the graph, was given in \cite{rahmani2009controllability}.  Controllability of MAS with switching network topology was studied in \cite{liu2008controllability}.  Although these studies characterized the controllability of leader-follower systems with a given leader set, they have  not addressed the question of finding the leader set.

A systematic framework for leader selection in the absence of noise was developed in \cite{liu2011controllability}.  The authors showed that leaders chosen according to a matching algorithm on the underlying communication graph of the MAS satisfy the structural controllability criterion~\cite{lin1974structural}, meaning that the system is controllable from the leaders for  any choice of follower dynamics, except in certain pathological cases.  The resulting polynomial-time algorithm returns the minimum number of leaders, as well as the identities of the leaders, needed to control the system.
This approach, however,  does not consider the effect of errors in the agent states that are introduced by noise in the communication links.

Leader selection in the presence of noise in networks with static topology was considered  in \cite{patterson2010leader}.  The authors of \cite{patterson2010leader} introduce the network coherence metric, which measures how close the follower agents states are to their desired consensus values, and equals the $H_{2}$-norm of the leader-follower system.  It is shown that the network coherence is a monotone nonincreasing function of the leader set, and a greedy algorithm for maximizing the network coherence is presented.  While the network coherence is equivalent to the metric we derive for static networks, provable bounds on the optimality of the selected leader sets cannot be derived from monotonicity alone. 



In~\cite{fardad2011noisefree,lin2011noisecorrupted}, the authors propose a semidefinite programming relaxation of the problem of selecting a set of up to $k$ leaders to minimize the $H_{2}$-norm defined in~\cite{patterson2010leader}.  These algorithms, however, do not provide any guarantees on the optimality of the chosen set of leaders.

Furthermore, while current approaches consider selecting a set of up to $k$ leaders in static networks in order to minimize the error in the agent states, the problem of selecting the minimum-size set of leader agents to meet a given bound on the error, as well as leader selection in dynamic networks, is not studied in~\cite{patterson2010leader,fardad2011noisefree,lin2011noisecorrupted}.  

When the leader set is given, the effect of noise on leader-follower MAS  protocols, such as consensus, has  been studied using a variety of approaches.   For a leader-follower system with additive link noise, the steady-state error due to noise was shown to be proportional to the graph effective resistance between the leader and follower agents in \cite{barooah2006graph}.  Alternative metrics, based on the transfer matrix between the noise inputs and the follower states, were proposed and analyzed in \cite{wang2009robustness}.
In~\cite{lin2011optimal}, decentralized control for vehicular networks with static topology, single- and double-integrator dynamics, and noise in the agent states was considered, and it was proved that at least one leader node must be present in the network to achieve stability.   In~\cite{bamieh2011coherence}, existing schemes for consensus and vehicle formation control were studied in the $H_{2}$-norm framework.  While these methods can be used to design and evaluate a leader-follower system with given leaders, they do not address the question of selecting a leader set in the presence of noise.     


The commute time of a random walk on a graph, defined as the expected time for a walk originating at a node $u$ to reach a node $v$, has been extensively studied~\cite{levin2009markov}.  In~\cite{chandra1996electrical}, it was shown that the graph effective resistance between two nodes $u$ and $v$ is proportional to the commute time of a random walk between $u$ and $v$.  To the best of our knowledge, however, our result that the commute time between a node $u$ and a set $S$ is a supermodular function of $S$ does not appear in the existing literature.
\section{Preliminaries and Background}
\label{sec:preliminaries}
\noindent In this section, we give preliminary information on the system model, system error, and supermodular functions.  
\subsection{System Model}
\label{subsec:system_model}
We consider a MAS consisting of $n$ agents, indexed by the set $V = \{1,\ldots, n\}$.  An edge $(i,j)$ exists between agents $i$ and $j$ iff $i$ and $j$ are within communication range of each other.  Letting $E$ denote the set of edges, the graph structure of the MAS is given by $G=(V,E)$.  For an agent indexed $i$, the neighbor set of $i$, denoted $N(i)$, is defined by $N(i) \triangleq \{j: (j,i) \in E\}$.  The degree of $i$ is defined to be the number of its neighbors $|N(i)|$.  It is assumed that the edges are undirected and the graph $G$ is connected.

Each agent $i$ has a time-varying state, denoted $x_{i}$(t), which may, in different contexts, represent $i$'s position, velocity, or sensed measurement.  Let $\mathbf{x} \in \mathbf{R}^{n}$ represent the vector of agent states.   The set of leader agents, denoted $S$, consists of agents who receive their state values directly from the MAS owner.  By broadcasting these state values to their one-hop neighbors, the leaders influence the dynamics of the follower agents.  Without loss of generality, we choose the indices such that $\mathbf{x}(t) = [\mathbf{x}_{f}(t)^{T} \quad \mathbf{x}_{l}(t)^{T}]^{T}$, where $\mathbf{x}_{f}(t)$ and $\mathbf{x}_{l}(t)$ denote the vectors of follower and leader states, respectively.

The goal of the MAS  is for the differences between the states of neighboring agents $i$ and $j$, $x_{i}-x_{j}$, to reach desired values, denoted $r_{ij}$, for all $(i,j) \in E$, so that $x_{i}-x_{j} = r_{ij}$.  The desired state $x_{i}^{\ast}$ is defined as the state satisfying $x_{i}^{\ast} - x_{j}^{\ast} = r_{ij}$ for all $j \in N(i)$.  We assume that $r_{ij}$ is known to agents $i$ and $j$ and the MAS owner, and that there exists at least one value of $\mathbf{x}^{\ast}$ such that $x_{i}^{\ast} - x_{j}^{\ast} = r_{ij}$ for all $(i,j) \in E$.

In the noiseless case, in order to reach the desired state $\mathbf{x}^{\ast}$, the follower agent $i \in V \setminus S$ updates its state according to the linear model
\begin{displaymath}
\dot{x}_{i}(t) = -\sum_{j \in N(i)}{W_{ij}(x_{i}(t) - x_{j}(t) - r_{ij})},
\end{displaymath}
where $W$ is a real-valued weight matrix with nonnegative entries. Furthermore, it is assumed that each link $(i,j)$ is affected by an additive, zero-mean white noise process, denoted $\epsilon_{ij}(t)$, with autocorrelation function $\mathbf{E}(\epsilon_{ij}(t)\epsilon_{ij}(t+\tau)) = \nu_{ij}\delta(\tau)$, where $\delta(\cdot)$ denotes the unit impulse function.  The noise values on each link are assumed to be independent.  This leads to the overall linear dynamics of agent $i \in V \setminus S$
\begin{equation}
\label{eq:general_linear_dynamics_noise}
\dot{x}_{i}(t) = -\sum_{j \in N(i)}{W_{ij}(x_{i}(t) - x_{j}(t) - r_{ij} + \epsilon_{ij}(t))}.
\end{equation}


In order to minimize the effect of noise, it is assumed that the link weights $W_{ij}$ are chosen in order to generate a best unbiased linear estimate of the leader agent states \cite{barooah2006graph}.  The link weights are therefore chosen as $W_{ij} = \nu_{ij}^{-1}/D_{i}$, where $D_{i} \triangleq \sum_{j \in N(i)}{\nu_{ij}^{-1}}$.  A detailed derivation of these dynamics is given in Appendix \ref{sec:derivation_agent_dynamics}.

Define the elements of the weighted Laplacian matrix $L$ by
\begin{equation}
\label{eq:Laplacian_def}
L_{ij} = \left\{
\begin{array}{cc}
-\nu_{ij}^{-1}, & (i,j) \in E \\
D_{i}, & i=j \\
0, & \mbox{else}
\end{array}
\right.
\end{equation}
which can be further decomposed as
 \begin{displaymath}
 L = \left(
 \begin{array}{c|c}
 L_{ff} & L_{fl} \\
 \hline
 L_{lf} & L_{ll}
 \end{array}
 \right),
 \end{displaymath}
 where $L_{ff}$ and $L_{fl}$ characterize the impact of the follower and leader agent states, respectively, on the follower update dynamics.
The dynamics of the follower agent (\ref{eq:general_linear_dynamics_noise}) for $i \in V \setminus S$ can be written in terms of $L$ as
 \begin{IEEEeqnarray*}{rCl}
 \dot{x}_{i}(t) &=& -D_{i}^{-1}\sum_{j \in N(i)}{\nu_{ij}^{-1}(x_{i}(t) - x_{j}(t) - r_{ij} + \epsilon_{ij}(t))} \\
  &=& -D_{i}^{-1}\left(\sum_{j \in N(i)}{\nu_{ij}^{-1}x_{i}(t)} - \sum_{j \in N(i)}{\nu_{ij}^{-1}x_{j}(t)} - \sum_{j \in N(i)}{\nu_{ij}^{-1}r_{ij}}\right) + w_{i}(t) \\
  &=& -D_{i}^{-1}\left(L_{ii}x_{i}(t) + \sum_{j \in N(i)}{L_{ij}x_{j}(t)} + \sum_{j \in N(i)}{L_{ij}r_{ij}}\right) + w_{i}(t),
  \end{IEEEeqnarray*}
  where $w_{i}(t)$ is a zero-mean white noise process.  Define matrix $B$ to be an $(n - |S|) \times |E|$ matrix, with $B_{il} = L_{ij}$ if edge $l$ is given by $(i,j)$ for some $j \in N(i)$ and $B_{il} = 0$ otherwise.
 The dynamics of the follower agents are given in vector form by
 \begin{displaymath}
\dot{\mathbf{x}}_{f}(t) = -D_{f}^{-1}(L_{ff}\mathbf{x}_{f}(t) + L_{fl}\mathbf{x}_{l}(t) + B\mathbf{r}) + \mathbf{w}(t),
\end{displaymath}
where $D_{f}$ is a $(n-|S|) \times (n-|S|)$ diagonal matrix with $D_{i}$ as the $i$-th diagonal entry.

We assume that the leaders maintain a constant state, $\mathbf{x}_{l}(t) \equiv \mathbf{x}_{l}^{\ast}$.  Based on this assumption, the desired state of the followers is defined to be $\mathbf{x}_{f}^{\ast} = -L_{ff}^{-1}(L_{fl}\mathbf{x}_{l}^{\ast} + B\mathbf{r})$, where we use the fact that $L_{ff}^{-1}$ exists when $G$ is connected~\cite[Lemma 10.36]{mesbahi2010graph}.

\subsection{Quantifying System Error}
\label{subsec:quantify_error}
The mean-square error in the follower agent states due to link noise in steady-state is defined as follows.
\begin{definition}
Let $\mathbf{x}_{f}^{\ast} \in \mathbf{R}^{n}$ denote a desired state for the follower agents. 
The total error of the follower agents at time $t$ is defined by $\mathbf{E}||\mathbf{x}_{f}(t)-\mathbf{x}_{f}^{\ast}||_{2}^{2}$. The system error, denoted $R(S)$, is defined as
\begin{displaymath}
R(S) \triangleq  \lim_{t \rightarrow \infty}{\mathbf{E}||\mathbf{x}_{f}(t)-\mathbf{x}_{f}^{\ast}||_{2}^{2}}.
\end{displaymath}
\end{definition}
The following theorem gives an explicit formula for $R(S)$ in terms of the matrix $L$.

\begin{theorem}
\label{theorem:error_Laplacian}
$R(S)$ is equal to $\frac{1}{2}Tr(L_{ff}^{-1}) = \frac{1}{2}\sum_{u \in V \setminus S}{(L_{ff}^{-1})_{uu}}$, where $(L_{ff}^{-1})_{uu}$ denotes the $(u,u)$-entry of $L_{ff}^{-1}$.
\end{theorem}

\begin{proof}
In the absence of noise, the follower agent dynamics are given by
\begin{equation}
\label{eq:noise_free_vector_dynamics}
\dot{\mathbf{x}}_{f}(t) =  -D_{f}^{-1}(L_{ff}\mathbf{x}_{f}(t) + L_{fl}\mathbf{x}_{l}^{\ast} + B\mathbf{r}).
\end{equation}
Since $L_{ff}$ is positive definite when $G$ is connected~\cite[Lemma 10.36]{mesbahi2010graph},  $\mathbf{x}_{f}^{\ast} = -L_{ff}^{-1}(L_{fl}\mathbf{x}_{l}^{\ast} + B\mathbf{r})$ is a global asymptotic equilibrium of (\ref{eq:noise_free_vector_dynamics}).  In the presence of noise, $\mathbf{x}_{f}(t)$ is given by
\begin{equation}
\label{eq:x_t_expression}
\mathbf{x}_{f}(t) = e^{-D_{f}^{-1}L_{ff}t}\mathbf{x}_{f}(0) - \int_{0}^{t}{e^{-D_{f}^{-1}L_{ff}(t-\tau)}D_{f}^{-1}(L_{fl}\mathbf{x}_{l}^{\ast} + B\mathbf{r})} +  \int_{0}^{t}{e^{-D_{f}^{-1}L_{ff}(t-\tau)}\mathbf{w}(\tau) \ d\tau}.
\end{equation}
 The first two terms converge to $\mathbf{x}_{f}^{\ast}$, since $\mathbf{x}_{f}^{\ast}$ is a global asymptotic equilibrium of (\ref{eq:noise_free_vector_dynamics}).  Since $\mathbf{w}$ is a zero-mean white process, the expected value of the third term of (\ref{eq:x_t_expression}) is zero by linearity of expectation.  Thus $\lim_{t \rightarrow \infty}{\mathbf{E}(\mathbf{x}_{f}(t)-\mathbf{x}_{f}^{\ast})} = 0$, leading to
 \begin{equation}
 \label{eq:variance}
R(S) = \lim_{t \rightarrow \infty}{\mathbf{E}||\mathbf{x}_{f}(t) - \mathbf{x}_{f}^{\ast}||_{2}^{2}} = \sum_{u \in V \setminus S}{\mathbf{E}((x_{u}(t)-x_{u}^{\ast})^{2})}
 =  \sum_{u \in V \setminus S}{var(x_{u}(t))},
 \end{equation}
where (\ref{eq:variance}) follows from the fact that $x_{u}(t) - x_{u}^{\ast}$ is zero-mean, where $x_{u}(t)$ denotes the state of agent $u$ at time $t$.  In \cite[Section IV-B]{barooah2006graph}, it was shown that $\lim_{t \rightarrow \infty}{var(x_{u}(t))} = \frac{1}{2}(L_{ff}^{-1})_{uu}$ (see Lemma \ref{lemma:appendix_error_proof} in Appendix \ref{sec:derivation_agent_dynamics} for a proof of this fact), implying that $R(S) = \frac{1}{2}\sum_{u \in V \setminus S}{(L_{ff}^{-1})_{uu}}$, as desired.
\end{proof}

For $u \in V \setminus S$, let $R(S,u) = \frac{1}{2}(L_{ff}^{-1})_{uu}$, so that $R(S) = \sum_{u \in V \setminus S}{R(S,u)}$.  Note that $L_{ff}^{-1}$ can be computed in worst-case $O(n^{3})$ time for a given leader set $S$.

\subsection{Supermodular Functions}
\label{subsec:submodular_background}
Let $V$ be a finite set, and let $2^{V}$ denote the set of subsets of $V$.  Then a supermodular function on $V$ is defined as follows.
\begin{definition}
\label{def:submodular}
Let $f: 2^{V} \rightarrow \mathbb{R}$, and let $S \subseteq T \subseteq V$.  The function $f$ is \emph{supermodular} if and only if, for any $v \in V \setminus T$,
\begin{equation}
\label{eq:submodular}
f(S)  - f(S \cup \{v\}) \geq  f(T) - f(T \cup \{v\}).
\end{equation}
\end{definition}
Intuitively, this identity implies that adding an element $v$ to a set $S$ results in a larger incremental decrease in $f$ than adding $v$ to a superset, $T$.  This can be interpreted as diminishing returns from $v$ as the set $S$ grows larger.
A function $f$ is \emph{submodular} if -$f$ is supermodular.  We first define the notion of a monotone set function.
\begin{definition}
\label{def:monotone}
A function $f: 2^{V} \rightarrow \mathbb{R}$ is \emph{monotone nondecreasing} (resp. \emph{nonincreasing}) if, for any $S \subseteq T$, $f(S) \leq f(T)$ (resp. $f(S) \geq f(T)$).
\end{definition}

The following two lemmas from \cite{fujishige2005submodular} will be used in our derivations below.
\begin{lemma}
\label{lemma:additive_submod}
Any nonnegative finite weighted sum of supermodular (resp. submodular) functions is supermodular (resp. submodular).
\end{lemma}
\begin{lemma}
\label{lemma:max_supermod}
Let $f: 2^{V} \rightarrow \mathbf{R}$ be a nonincreasing supermodular function, and let $c \geq 0$ be a constant.  Then $g(S) = \max{\{f(S), c\}}$ is supermodular.
\end{lemma}
A proof of Lemma \ref{lemma:max_supermod} can be found in Appendix \ref{sec:supermod_lemma_proof}.

\section{Leader Selection in Static Networks}
\label{sec:problem_static}
\noindent In this section, we consider the problem of selecting a leader set in order to minimize the system error in static networks, in which the set of links $E$ and the error variances $\nu_{ij}$ do not change over time.  We address this problem for two cases.  In the first case, no more than $k$ agents can act as leaders, and the goal is to choose a set of leaders that minimizes the system error.  In the second case, the system error cannot exceed an upper bound $\alpha$, and the goal is to find the minimum number of leaders, as well as the identities of the leader, such that the system error is less than or equal to $\alpha$.  In both cases, we construct algorithms for leader selection.
\subsection{Case I -- Choosing up to $k$ Leaders to Minimize System Error}
\label{subsec:primal_formulation}

The problem of choosing a set $S$ of $k$ leaders in order to minimize the system error $R(S)$ is given by
\begin{equation}
\label{eq:primal_static_opt}
\begin{array}{cc}
\mbox{minimize} & R(S) \\
\mbox{s.t.} & |S| \leq k
\end{array}
\end{equation}


Note that problem (\ref{eq:primal_static_opt}) always has at least one feasible solution when $k \geq 1$, for example any set consisting of a single node.  In what follows, we prove that the total system error $R(S)$ is a supermodular function of the leader set $S$, leading to efficient algorithms for approximating the optimal solution to (\ref{eq:primal_static_opt}) up to a provable bound.  We first prove that for each agent $u$, $R(S,u)$ is proportional to the commute time of a random walk on the graph $G$ from agent $u$ to the leader set $S$, and then show that the commute time is supermodular.  The supermodularity of $R(S)$ follows as a corollary.





\begin{theorem}
\label{theorem:commute_time}
Define a random walk on the graph $G$ starting at node $u \in V \setminus S$, in which the probability of transitioning from node $i$ to node $j \in N(i)$, denoted $P(i,j)$, is given by
\begin{equation}
\label{eq:probability_distribution}
P(i,j) = \frac{\nu_{ij}^{-1}}{D_{i}}.
\end{equation}
The commute time $\kappa(S,u)$, defined as the expected number of steps for the random walk to reach the leader set $S$ and return to $u$, is proportional to $R(S,u) = (L_{ff}^{-1})_{uu}$.
\end{theorem}
The proof of Theorem \ref{theorem:commute_time} can be found in the appendix.  

\begin{theorem}
\label{theorem:commute_time_supermodular}
The commute time $\kappa(S,u)$ is a nonincreasing supermodular function of $S$.
\end{theorem}

\begin{proof}
The nonincreasing property follows from the fact that, if $S \subseteq T$, then any walk that reaches $S$ and returns to $u$ has also reached $T$ and returned to $u$.  Thus $\kappa(T,u) \leq \kappa(S,u)$.

By Definition \ref{def:submodular}, $\kappa(S,u)$ is a supermodular function of $S$ if and only if, for any sets $A$ and $B$ with $A \subseteq B$ and for any $j \notin B$,
\begin{equation}
\label{eq:commute_supermod}
\kappa(A,u) - \kappa(A \cup \{j\}, u) \geq \kappa(B,u) - \kappa(B \cup \{j\},u).
\end{equation}
Consider the quantity $\kappa(A,u) - \kappa(A \cup \{j\},u)$.  Define $A^{\prime} = A \cup \{j\}$, and define $T_{Au}$ and $T_{A^{\prime}u}$ to be the (random) times for a random walk to reach $A$ (respectively $A^{\prime}$) and return to $u$.  Then by definition, $\kappa(A,u) = \mathbf{E}(T_{Au})$ and $\kappa(A^{\prime},u) = \mathbf{E}(T_{A^{\prime}u})$.  This implies $\kappa(A,u) - \kappa(A^{\prime},u) = \mathbf{E}(T_{Au} - T_{A^{\prime}u})$.

Let $\tau_{j}(A)$ denote the event where the random walk reaches node $j$ before any of the nodes in $A$.  Further, let $h_{jAu}$ be the  time for the random walk to travel from $j$ to $A$ and then to $u$, while $h_{ju}$ is the time to travel directly from $j$ to $u$.  We have
\begin{IEEEeqnarray*}{rCl}
\nonumber
\kappa(A,u) - \kappa(A^{\prime},u) &=& \mathbf{E}(T_{Au} - T_{A^{\prime}u} | \tau_{j}(A))Pr(\tau_{j}(A)) + \mathbf{E}(T_{Au} - T_{A^{\prime}u}|\tau_{j}(A)^{c})Pr(\tau_{j}(A)^{c}) \\
  &=& \mathbf{E}(h_{jAu} - h_{ju})Pr(\tau_{j}(A))
  \end{IEEEeqnarray*}
noting that if the walk reaches $A$ first, then $T_{Au}$ and $T_{A^{\prime}u}$ are equal.  Hence (\ref{eq:commute_supermod}) becomes
\begin{equation}
\label{eq:commute_supermod_equiv}
\mathbf{E}(h_{jAu} - h_{ju})Pr(\tau_{j}(A)) \geq \mathbf{E}(h_{jBu} - h_{ju})Pr(\tau_{j}(B))
\end{equation}
In order to prove (\ref{eq:commute_supermod_equiv}) holds, it suffices to prove that $h_{jAu} \geq h_{jBu}$ and $Pr(\tau_{j}(A)) \geq Pr(\tau_{j}(B))$.  Let $W_{jau}^{k}$ denote the event that, after $k$ steps, a random walk starting at $j$ has either not reached node $a \in A$, or has reached $a$ but has not yet reached $u$. Define $W_{jAu}^{k} = \cap_{a \in A}{W_{jau}^{k}}$.  Let $I(W_{jAu}^{k})$ denote the indicator function of the set $W_{jAu}$, and note that
 \begin{equation}
 I(W_{jAu}^{k}) = \prod_{a \in A}{I(W_{jau}^{k})}
 \end{equation}
 by the observations above and the definition of indicator functions.
 $h_{jAu}$ can  be rewritten as
  \begin{IEEEeqnarray*}{rCl}
 h_{jAu} &=& \min{\{k: I(W_{jAu}^{k}) = 0\}} =: k^{\ast} \\
  \label{eq:first_step}
  &\stackrel{(a)}{=}& \sum_{k=0}^{k^{\ast}-1}{I(W_{jAu}^{k})}
  \stackrel{(b)}{=} \sum_{k=0}^{k^{\ast}-1}{I(W_{jAu}^{k})} + \sum_{k=k^{\ast}}^{\infty}{I(W_{jAu}^{k})} \\
  &=& \sum_{k=0}^{\infty}{\prod_{a \in A}{I(W_{jau}^{k})}}
  \stackrel{(c)}{\geq} \sum_{k=0}^{\infty}{\left(\prod_{a \in A}{I(W_{jau}^{k})}\prod_{b \in B \setminus A}{I(W_{jbu}^{k})}\right)} \\
  &=& \sum_{k=0}^{\infty}{\prod_{b \in B}{I(W_{jbu}^{k})}} = \sum_{k=0}^{\infty}{I(W_{jBu}^{k})}= h_{jBu}
 \end{IEEEeqnarray*}
 where (a) follows from the definition of $I(W_{jAu}^{k})$, (b) follows from the fact that $I(W_{jAu}^{k}) = 0$ for $k \geq k^{\ast}$, and  (c) follows from the fact that $\left(\bigcap_{a \in A}{W_{jau}^{k}}\right) \cap \left(\bigcap_{b \in B \setminus A}{W_{jbu}^{k}}\right) \subseteq \bigcap_{a \in A}{W_{jau}^{k}}$.

 In order to show that $Pr(\tau_{j}(A)) \geq Pr(\tau_{j}(B))$, first let $\tau_{j}(v)$ denote the event that a random walk reaches node $j \neq v$ before $v$. Hence,
 \begin{displaymath}
 \tau_{j}(B) = \bigcap_{v \in B}{\tau_{j}(v)}
  = \left(\bigcap_{v \in A}{\tau_{j}(v)}\right) \cap \left(\bigcap_{v \in B\setminus A}{\tau_{j}(v)}\right) \subseteq \bigcap_{v \in A}{\tau_{j}(v)} = \tau_{j}(A)
  \end{displaymath}

  $\tau_{j}(B) \subseteq \tau_{j}(A)$ then implies that $Pr(\tau_{j}(B)) \leq Pr(\tau_{j}(A))$, as desired.
 Hence (\ref{eq:commute_supermod_equiv}) holds, thus proving that the commute time is supermodular as a function of the leader set $S$.
\end{proof}

An alternate proof of Theorem \ref{theorem:commute_time_supermodular} can be found in Appendix \ref{sec:alternate_proof_main_theorem}.


\begin{corollary}
\label{theorem:R_supermodular}
$R(S)$ is a nonincreasing supermodular function of the leader set, $S$.
\end{corollary}

\begin{proof}
By Theorem \ref{theorem:commute_time_supermodular}, $\kappa(S,u)$ is supermodular as a function of $S$.  Since $R(S,u) = (L_{f}^{-1})_{uu}$ is proportional to $\kappa(S,u)$, $(L_{f}^{-1})_{uu}$ is supermodular as a function of $S$ as well.  By Theorem \ref{theorem:error_Laplacian}, $R(S) = \sum_{u \in V \setminus S}{(L_{f}^{-1})_{uu}}$,  hence $R(S)$ is a sum of supermodular functions.  $R(S)$ is therefore supermodular by Lemma \ref{lemma:additive_submod}.
\end{proof}

 Corollary \ref{theorem:R_supermodular} implies that the problem of selecting up to $k$ leader agents in order to minimize the system error for a static network (\ref{eq:primal_static_opt}) is a supermodular optimization problem.
Although supermodular optimization problems of this form are NP-hard in general, a greedy algorithm will return a set $S^{\ast}$ such that $R(S^{\ast})$ is within a factor of $(1-1/e)$ of the optimum value, denoted $R^{\ast}$~\cite{nemhauser1978analysis}.

We now present a greedy algorithm for selecting a set of $k$ leaders for a static network topology, which is an approximate solution to (\ref{eq:primal_static_opt}).
Let $S_{i}^{\ast}$ denote the set of leader agents at the $i$-th iteration of the algorithm.  $S_{0}^{\ast}$ is initialized to $\emptyset$.  At the $i$-th iteration of the algorithm, the element $s_{i}^{\ast} \in V$ is found such that $\{R(S_{i-1}^{\ast}) - R(S_{i-1}^{\ast} \cup \{s_{i}^{\ast}\})\}$ is maximized.  $S_{i}^{\ast}$ is then updated to $(S_{i-1}^{\ast} \cup \{s_{i}^{\ast}\})$.  The algorithm terminates when either $R(S_{i}^{\ast}) = R(S_{i}^{\ast} \cup \{j\})$ for all $j$, or when $i=k$ (i.e., when the number of leaders is equal to $k$), whichever condition is reached first.  A pseudocode description of the algorithm is given as algorithm \textbf{static}-$k$.

\begin{figure}[ht]
\centering
\renewcommand{\arraystretch}{0.6}
\begin{tabular}{lr}
\hline
\small{\textbf{Algorithm static-$k$}: Algorithm for choosing up to $k$ leaders }& \\
\hline
\small{\textbf{Input:} $G=(V,E)$, link error variances $\nu_{ij}$}  & \\
~~~~~~~\small{Maximum number of leader nodes $k$} &\\
\small{\textbf{Output:} Set of leader nodes $S^{\ast}$} &\\
\small{\textbf{Initialization:} $S^{\ast} \leftarrow \emptyset$, $i \leftarrow 0$} & \small{1}\\
\small{\textbf{while} $i \leq k$} & \small{2}\\
\small{~~~~$s_{i}^{\ast} \leftarrow \arg\max_{j \in V \setminus S}{\{R(S^{\ast}) - R(S^{\ast} \cup \{j\})\}}$} & \small{3}\\
~~~~\small{\textbf{if} $R(S^{\ast}) - R(S^{\ast} \cup \{s_{i}^{\ast}\})\} \leq  0$} & \small{4}\\
~~~~~~~~\small{\textbf{return} $S^{\ast}$; \textbf{exit}} & \small{5}\\
~~~~\small{\textbf{else}} & \small{6}\\
~~~~~~~~\small{$S^{\ast} \leftarrow S^{\ast} \cup \{s_{i}^{\ast}\}$} & \small{7}\\
~~~~~~~~\small{$i \leftarrow i+1$} & \small{8}\\
~~~~\small{\textbf{end}} & \small{9}\\
\small{\textbf{end}} & \small{10}\\
\small{\textbf{return} $S^{\ast}$; \textbf{exit}} & \small{11}\\
\hline
\end{tabular}
\label{algorithm:primal_static}
\end{figure}

The following theorem gives a bound on the performance of \textbf{static-$k$}, making use of the fact that (\ref{eq:primal_static_opt}) is a supermodular optimization problem.

\begin{theorem}
\label{theorem:primal_static_opt_bound}
Define $R_{max}$ by
$R_{max} \triangleq \max_{i}{R(\{i\})}$,
which is the worst-case error when a single agent is chosen as a leader, and let $R^{\ast}$ be the optimal value of (\ref{eq:primal_static_opt}).
Then the algorithm \textbf{static}-$k$ terminates in polynomial time and returns a set $S^{\ast}$ satisfying
\begin{displaymath}
R(S^{\ast}) \leq \left(1 - \left(\frac{k-1}{k}\right)^{k}\right)R^{\ast} + \frac{1}{e}R_{max}
\approx \left(1 - \frac{1}{e}\right)R^{\ast} + \frac{1}{e}R_{max}
\end{displaymath}
\end{theorem}

\begin{proof}
By Theorem 9.3 of \cite[Ch III.3.9]{wolsey1999integer}, for any nonnegative monotone nondecreasing submodular function $f(S)$, the greedy algorithm returns a set $S^{\ast}$ satisfying $f(S^{\ast}) \geq (1 - 1/e)f^{\ast}$, where $f^{\ast}$ is the optimal value of the optimization problem
\begin{equation}
\label{eq:general_submod}
\begin{array}{cc}
\mbox{maximize} & f(S) \\
\mbox{s.t.} & |S| \leq k
\end{array}
\end{equation}
Since $R(S)$ is nonincreasing and supermodular by Corollary \ref{theorem:R_supermodular}, the function $f(S) = R_{max} - R(S)$ is nonnegative,
nondecreasing, and submodular.  Maximizing $f(S)$ is thus equivalent to minimizing $R(S)$.  Hence, the set $S^{\ast}$ returned by the greedy algorithm satisfies
$f(S^{\ast}) \geq \left(1 - \frac{1}{e}\right)f^{\ast}$.
Substituting the definition of $f(S)$ yields
$R_{max} - R(S^{\ast}) \geq \left(1 - \frac{1}{e}\right)(R_{max} - R^{\ast})$
and rearranging terms gives the desired result.

Algorithm \textbf{static}-$k$ requires $k$ iterations.  At each iteration, the function $R(S)$ is evaluated $O(n)$ times. Since each evaluation of $R(S)$ involves inverting the Laplacian matrix, which requires $O(n^{3})$ operations (and can be reduced to $O(n^{2})$ operations, with some loss in computation accuracy, if the Laplacian matrix is sparse~\cite{patterson2010leader}), the total runtime is $O(kn^{4})$.
\end{proof}

In \cite{nemhauser1978best}, it was shown that no polynomial-time algorithm that improves on the approximation bound $\left(1 - \frac{1}{e}\right)$ for an arbitrary problem of the for (\ref{eq:general_submod}) unless $P=NP$.  There may, however, be additional structure to $R(S)$ other than supermodularity that may improve the guarantees of Theorem \ref{theorem:primal_static_opt_bound}.  This remains an open problem.


\subsection{Case II -- Choosing the Minimum-Size Leader Set to Achieve an Error Bound}
\label{subsec:static_dual}

When the system is required to operate below a given error bound, denoted $\alpha$, the problem of choosing a minimal set of leaders that achieves this bound can be stated as

\begin{equation}
\label{eq:dual_opt}
\begin{array}{cc}
\mbox{minimize} & |S| \\
\mbox{s.t.} & R(S) \leq \alpha
\end{array}
\end{equation}

Note that, for any $\alpha \geq 0$, there exists at least one $S$ meeting the condition $R(S) \leq \alpha$, namely the set $S = V$.

The supermodularity of $R(S)$ enables an efficient approximate solution of (\ref{eq:dual_opt}) by a greedy algorithm, given as follows.  The set of leaders is initialized to $S_{0}^{\ast} = \emptyset$.  As with the leader selection algorithm \textbf{static}-$k$, the node $s_{i}^{\ast}$ that maximizes $\{R(S_{i-1}^{\ast}) - R(S_{i-1}^{\ast} \cup \{s_{i}^{\ast}\})\}$ is added at the $i$-th iteration, so that $S_{i}^{\ast} = S_{i-1}^{\ast} \cup \{s_{i}^{\ast}\}$.  The algorithm terminates when $R(S_{i}^{\ast}) \leq \alpha$ and returns the set $S^{\ast} = S_{i}^{\ast}$.  A pseudocode description of the algorithm is given as algorithm \textbf{static}-$\alpha$.


\begin{figure}[ht]
\centering
\renewcommand{\arraystretch}{0.6}
\begin{tabular}{lr}
\hline
\small{\textbf{Algorithm static-$\alpha$}: Algorithm for choosing the minimum-size} &\\
\small{set of leaders to achieve an error bound $\alpha$} & \\
\hline
\small{\textbf{Input:} $G=(V,E)$, link error variances $\nu_{ij}$} & \\
~~~~~~~\small{Error at termination $\alpha$} & \\
\small{\textbf{Output:} Set of leader nodes $S^{\ast}$} & \\
\small{\textbf{Initialization:} $S^{\ast} \leftarrow \emptyset$, $error \leftarrow \alpha+1$} & \small{1}\\
\small{\textbf{while} $error > \alpha$} & \small{2}\\
~~~~\small{$s^{\ast} \leftarrow \arg\max_{j \in V \setminus S}{R(S^{\ast}) - R(S^{\ast} \cup \{j\})\}}$} & \small{3}\\
~~~~\small{\textbf{if} $(R(S^{\ast}) - R(S^{\ast} \cup \{s^{\ast}\})) \leq 0$} & \small{4}\\
~~~~~~~~\small{\textbf{return} $S^{\ast}$; \textbf{exit}} & \small{5}\\
~~~~\small{\textbf{else}} & \small{6}\\
~~~~~~~~\small{$S^{\ast} \leftarrow S^{\ast} \cup \{s^{\ast}\}$} & \small{7}\\
~~~~~~~~\small{$error \leftarrow R(S^{\ast})$} & \small{8}\\
~~~~\small{\textbf{end}} & \small{9}\\
\small{\textbf{end}} & \small{10}\\
\small{\textbf{return} $S^{\ast}$; \textbf{exit}} & \small{11}\\
\hline
\end{tabular}
\label{algorithm:dual_static}
\end{figure}

The following theorem gives bounds on the optimality of the set $S^{\ast}$ returned by \textbf{static-$\alpha$}.
\begin{theorem}
\label{theorem:dual_bound}
Let $k^{\ast}$ be the smallest integer such that a set $S$ exists with $|S| = k$ and $R(S) \leq \alpha$ (i.e., $k^{\ast}$ is the optimal value of (\ref{eq:dual_opt})).
 The algorithm \textbf{static}-$\alpha$ terminates in polynomial time in $n$.  If the algorithm terminates after step $k$, so that $|S^{\ast}| = k$, then
\begin{displaymath}
\frac{k}{k^{\ast}} \leq 1 + \log{\left\{ \frac{R_{max}}{R(S_{k-1})}\right\}}
\end{displaymath}
holds, where
$R_{max}$ is as defined in Theorem \ref{theorem:primal_static_opt_bound}.
\end{theorem}
\begin{proof}
Theorem 9.4 of \cite[Ch III.3.9]{wolsey1999integer} states that the greedy algorithm for solving problems of the form
\begin{equation}
\label{eq:generic_dual_form}
\begin{array}{cc}
\mbox{minimize}_{ \ S \subseteq V} & |S| \\
\mbox{s.t.} & f(S) \geq \lambda
\end{array}
\end{equation}
where $f(S)$ is a nondecreasing submodular function returns a set $S^{\ast}$ satisfying
\begin{displaymath}
\frac{|S^{\ast}|}{|S^{\prime}|} \leq 1 + \log{\left\{\frac{f(V) - f(\emptyset)}{f(V) - f(S_{k-1})}\right\}},
\end{displaymath}
where $S^{\prime}$ is the optimal solution to (\ref{eq:generic_dual_form}) and $S_{k-1}$ is the set obtained at the $(k-1)$-th iteration of the greedy algorithm.  Letting $f(S) = R_{max} - R(S)$, we have that $f(S)$ is a nondecreasing submodular function of $S$.  Since the greedy algorithm for optimizing $f(S)$ is equivalent to optimizing $R(S)$, we have
\begin{displaymath}
\frac{|S^{\ast}|}{|S^{\prime}|} \leq 1 + \log{\left\{\frac{f(V) - f(\emptyset)}{f(V) - f(S_{k-1})}\right\}}
 = 1 + \log{\left\{\frac{-f(\emptyset)}{-f(S_{k-1})}\right\}}
 = 1 + \log{\left\{\frac{R_{max}}{f(S_{k-1})}\right\}}
 \end{displaymath}
In the worst case, the algorithm will not terminate until $S = V$, i.e., after $n$ iterations.  This will require $O(n^{2})$ evaluations of $R$, each of which requires $O(n^{3})$ computations, for a total runtime of $O(n^{5})$.
\end{proof} 
\section{Leader Selection in Dynamic Networks}
\label{sec:problem_dynamic}
\noindent Multi-agent systems may undergo changes in topology or link noise characteristics for three reasons.  First, the agents may experience link or device failures~\cite{hatano2005agreement}.  Second, the MAS may switch between prespecified topologies~\cite{olfati2004consensus}.  Third, the topology may vary arbitrarily over time due to agent mobility; for example, agents may change positions in order to avoid an obstacle, affecting both the set of links $E$ and the set of error variances $\nu_{ij}$~\cite{olfati2006flocking}.  Under each case, the optimal set of leaders may be different from that of a static network.  In this section, we study leader selection for each of these dynamic networks.
\subsection{Leader Selection Under Random Link Failures}
\label{subsec:random_topology}
Random topology changes may occur due to link failures, which are reflected in the weighted Laplacian matrix $L$ of (\ref{eq:Laplacian_def}).
Since the set of links may not be known in advance under these circumstances, leaders can  be selected to minimize the expected system error based on  the distribution of possible weighted Laplacians.

Let $\mathcal{L}$ denote the set of possible weighted Laplacians.  Define $\pi$ to be a probability distribution on $\mathcal{L}$, so that $\pi(L)$ is  the probability that the Laplacian is $L \in \mathcal{L}$.  An example distribution is the \emph{random link failure} model, in which each link in an underlying link set $E$ fails independently with equal probability $p$.  The expected system error is defined by
\begin{displaymath}
\mathbf{E}_{\pi}(R(S)) = \sum_{L \in \mathcal{L}}{R(S|L)\pi(L)},
\end{displaymath}
where $R(S|L)$ denotes the system error when the leader set is $S$ and the Laplacian is $L$.

The problems of (a) choosing a set of $k$ leaders to minimize the expected error $\mathbf{E}_{\pi}(R(S))$ and (b) choosing the smallest possible leader set $S$ such that the error is within an upper bound $\alpha$ are formulated as
\begin{equation}
\label{eq:random_failures_opt}
\begin{array}{c|c}
\mbox{\textbf{Choosing up to $k$ leaders}} & \mbox{\textbf{Choosing minimum number of leaders}} \\
\begin{array}{ccl}
\mbox{minimize} & & \mathbf{E}_{\pi}(R(S)) \\
\mbox{s.t.} & & |S| \leq k
\end{array}
&
\begin{array}{ccl}
\mbox{minimize} & & |S| \\
\mbox{s.t.} & & \mathbf{E}_{\pi}(R(S)) \leq \alpha
\end{array}
 \\
(a) & (b)
\end{array}
\end{equation}

In order to solve (\ref{eq:random_failures_opt}a) and (\ref{eq:random_failures_opt}b), it is necessary to compute $\mathbf{E}_{\pi}(R(S))$ for a given distribution $\pi$ and leader set $S$. The summation, however, can have up to $2^{|E|}$ possible topologies under the random link failure model, making exact computation of $\mathbf{E}_{\pi}(R(S))$  difficult.  We present two approaches to approximating $\mathbf{E}_{\pi}(R(S))$: first, a Monte Carlo approximation that is valid for any distribution $\pi$, and second, a gradient-based approximation that is valid when the probability of link failures is small.

\noindent \underline{\emph{Monte Carlo Approximation:}}
Under the Monte Carlo approach, a set consisting of $M$ Laplacian matrices $\{L_{1}, \ldots, L_{M}\}$, each chosen independently according to distribution $\pi$, is generated.
$\mathbf{E}_{\pi}(R(S))$ is then approximated by $R_{mc}(S) = \frac{1}{M}\sum_{i=1}^{M}{R(S|L_{i})}$.
\begin{theorem}
\label{theorem:MC_convergence}
For any $\epsilon > 0$,
$\lim_{M \rightarrow \infty}{Pr(|R_{mc}(S) - \mathbf{E}_{\pi}(R(S))| < \epsilon)} = 1$
for every $S \subseteq V$.
\end{theorem}
\begin{proof}
For a given $S$, each $R(S|L_{i})$ represents an independent sample from the probability distribution $\pi$.  Hence, by the weak law of large numbers, $\sum_{i=1}^{M}{R(S|L_{i})}$ converges in probability to the expected value $\mathbf{E}_{\pi}(R(S))$.
\end{proof}

\noindent \underline{\emph{Gradient Approximation:}}
Consider the random link failure model, and let $\tilde{G} = (V,\tilde{E})$, where $\tilde{E}$ denotes the set of links that have failed.  Define matrix $\Delta$ by
\begin{displaymath}
\Delta_{ij} = \left\{
\begin{array}{cc}
-\nu_{ij}^{-1}, & (i,j) \in \tilde{E} \\
\sum_{(i,j) \in \tilde{E}}{\nu_{ij}^{-1}}, & i=j \\
0, & \mbox{else}
\end{array}
\right.
\end{displaymath}
Let $\tilde{L}_{ff} \triangleq L_{f} - \Delta$, so that $R(S|\tilde{L}) = \sum_{u \in V \setminus S}{(\tilde{L}_{ff}^{-1})_{uu}}$.

\begin{lemma}
\label{lemma:random_gradient_helper}
Suppose that links fail independently and randomly with probability $p$.  Let $X = \max_{(i,j) \in E}{\nu_{ij}^{-1}}$, and let $d$ denote the maximum node degree of the graph $G$.  Define $||\Delta||_{2}$ to be the maximum singular value of $\Delta$. Then $||\Delta||_{2} \leq 2pdX$.
\end{lemma}

\begin{proof}
Since $\Delta$ is symmetric and positive semidefinite, $||\Delta||_{2}$ is equal to the maximum eigenvalue of $\Delta$.
Let $E_{i}$ denote the set of edges incident on agent $i$ which fail.  Then
\begin{equation}
\sum_{j \neq i}{|\Delta_{ij}|} = \sum_{j \in E_{i}}{\nu_{ij}^{-1}} \leq X \sum_{j \in E_{i}}{1} = X|E_{i}| \leq pdX,
\end{equation}
for each $i \in V$.
Furthermore, $\Delta_{ii} = \sum_{j \neq i}{|\Delta_{ij}|} \approx pdX$.  Hence by the Gershgorin Disc Theorem \cite{trefethen1997numerical}, the eigenvalues of $\Delta$ must lie in the interval $[0, 2pdX]$.
\end{proof}

Lemma \ref{lemma:random_gradient_helper} leads to the following gradient approximation for $\mathbf{E}_{\pi}(R(S))$.

\begin{theorem}
\label{theorem:random_bound}
Let $p$, $d$, and $X$ be as defined above, and let $\delta = 2pdX$.  Then
\begin{equation}
R(S | \tilde{G}) = Tr(\tilde{L}_{ff}^{-1}) \leq Tr(L_{ff}^{-1}) + \frac{(n-|S|)\delta}{\lambda_{min}(L_{ff})^{2}} ,
\end{equation}
where $\lambda_{min}(L_{ff})$ is the smallest eigenvalue of $L_{f}$.
\end{theorem}
\begin{IEEEproof}
$R(S)$ can be written as
\begin{IEEEeqnarray*}{rCl}
Tr((L_{ff} - \Delta)^{-1}) &\approx& Tr(L_{ff}^{-1}) + Tr(L_{ff}^{-1}\Delta L_{ff}^{-1}) \\
 &\leq& Tr(L_{ff}^{-1}) + Tr(U\Lambda^{-1}U^{T}\Delta U\Lambda^{-1}U^{T}) \nonumber\\
 &\leq& Tr(L_{ff}^{-1}) + \sup_{\Delta}{Tr(U\Lambda^{-1}U^{T}\Delta U\Lambda^{-1}U^{T})} \nonumber
\end{IEEEeqnarray*}
where $L_{f} = U\Lambda U^{T}$ is the eigen-decomposition of $L_{f}$.  The upper bound occurs when $\Delta = U\Omega U^{T}$ for some positive semidefinite diagonal matrix $\Omega$.
Together with Lemma \ref{lemma:random_gradient_helper}, this implies that
\begin{IEEEeqnarray*}{rCl}
Tr((L_{f}-\Delta)^{-1}) &\leq& Tr(L_{ff}^{-1}) + Tr(U\Lambda^{-1}\Omega\Lambda^{-1}U^{T}) \\
&=& Tr(L_{ff}^{-1}) Tr(U^{T}U\Lambda^{-1}\Omega\Lambda^{-1}) = Tr(L_{ff}^{-1}) + Tr(\Lambda^{-1}\Omega\Lambda^{-1}) \\
&\leq& Tr(L_{ff}^{-1}) + \frac{(n-|S|)\delta}{\lambda_{min}^{2}}.
\end{IEEEeqnarray*}
\end{IEEEproof}

Note that the upper bound on $R(S)$ can  be used as a worst-case value for $R(S)$ in the presence of link failures. 
Once an appropriate method for computing or estimating $\mathbf{E}_{\pi}(R(S))$ has been chosen, the leader set $S$ can be selected by solving (\ref{eq:random_failures_opt}a) or (\ref{eq:random_failures_opt}b).
The following lemma can be used to derive efficient algorithms for both problems.




\begin{lemma}
\label{lemma:random}
Under the random link failure model, for any distribution $\pi$ on $\mathcal{L}$, the function $\mathbf{E}_{\pi}(R(S))$ is supermodular.
\end{lemma}
\begin{proof}
By definition,
$\mathbf{E}_{\pi}(R(S)) = \sum_{L \in \mathcal{L}}{R(S|L)\pi(L)}$.
Since the set $\mathcal{L}$ of possible weighted Laplacians is finite, this is a nonnegative weighted sum of supermodular functions, and hence is supermodular by Lemma \ref{lemma:additive_submod}.
\end{proof}

As a corollary to Lemma \ref{lemma:random}, the problem of choosing a set of $k$ leaders to minimize the expected error (\ref{eq:random_failures_opt}a), and the problem of selecting the minimum number of leaders to meet a bound on the expected error (\ref{eq:random_failures_opt}b), can be solved using algorithms analogous to \textbf{static-$k$} and \textbf{static-$\alpha$} 
respectively.  The modified algorithms take the distribution $\pi$ as an additional input parameter, and replace line three in both algorithms with $s_{i}^{\ast} \leftarrow \arg\max_{j \in V \setminus S}{\{\mathbf{E}_{\pi}(R(S^{\ast})) - \mathbf{E}_{\pi}(R(S^{\ast} \cup \{j\}))\}}$.

\subsection{Leader Selection Under Switching Between Predefined Topologies}
\label{subsec:switching}

A MAS may switch between a set of predefined topologies $\mathcal{T} = \{G_{1}, \ldots, G_{M}\}$, each with different link error variances represented by the corresponding weighted Laplacians $L_{1}, \ldots, L_{M}$ (e.g., a set of possible formations) in response to a switching signal from the MAS owner or environmental changes~\cite{olfati2004consensus}.
Leader selection under switching topologies can be divided into two cases.
In the first case, the set of leaders is updated after each change in topology~\cite{mesbahi_spacecraft}.  For this case, a different set of leaders $S_{i}$ can be selected for each topology $G_{i}$ using either \textbf{static-$k$} (if a fixed number $k$ of leaders is chosen to minimize error) or \textbf{static-$\alpha$} (if the minimum number of leaders is chosen to achieve an error bound $\alpha$).


In the second case, the same leader set $S$ is chosen and used for all topologies~\cite{liu2008controllability}.
Under this strategy, we consider two possible leader selection metrics, namely the average-case error, given as $R_{avg}(S) = \frac{1}{M}\sum_{i=1}^{M}{R(S|L_{i})}$, and the worst-case error, given as $R_{worst}(S) = \max{\{R(S|L_{i}) : i=1,\ldots, M\}}$.  $R_{avg}(S)$ is a nonnegative weighted sum of supermodular functions, and hence is supermodular as a function of $S$ by Lemma \ref{lemma:additive_submod}.  The problems of selecting up to $k$ leaders in order to minimize $R_{avg}(S)$ and selecting the minimum number of leaders to achieve an error bound $\alpha$ can  therefore be solved by  modified versions of \textbf{static-$k$} and \textbf{static-$\alpha$}, respectively.  The modified versions of both leader selection algorithms take the topologies $\{G_{1}, \ldots G_{M}\}$ as input, and replace line 3 in both algorithms with $s_{i}^{\ast} \leftarrow \arg\max_{j \in V \setminus S}{\left\{R_{avg}(S^{\ast})\right.}$ $\left. - R_{avg}(S^{\ast} \cup \{j\})\right\}$.  Similarly, the problem of selecting the minimum number of leaders to achieve an error bound $\alpha$ can be solved by a modified version of \textbf{static-$\alpha$}, also taking $\{G_{1}, \ldots, G_{M}\}$ as input, and with the same substitution at line 3.

If $R_{worst}(S)$ is used as a metric, however,  note that the maximum of supermodular functions is, in general, not supermodular~\cite{Krause_adversarial}.  Alternate approaches for leader selection problems are given as follows.

\subsubsection{Choosing $k$ leaders to minimize worst-case error}
\label{subsubsec:primal_switching}
The problem of choosing $k$ leaders in order to minimize $R_{worst}$ is stated as
\begin{equation}
\label{eq:primal_switching}
\begin{array}{cc}
\mbox{minimize} & R_{worst}(S) \triangleq \max_{i=1,\ldots,M}{R(S|L_{i})} \\
\mbox{s.t.} & |S| \leq k
\end{array}
\end{equation}

Let $S^{\ast}$ be the solution to (\ref{eq:primal_switching}), and let $R_{worst}^{\ast} = R_{worst}(S^{\ast})$.  $R_{worst}^{\ast}$ is bounded below by $0$ and bounded above by $R_{max} = \max_{i=1,\ldots, M}{\max_{u \in V}{\{R(\{u\}|L_{i})\}}}$.  As a preliminary, define
\begin{equation}
\label{eq:def_F}
F_{c}(S) \triangleq \frac{1}{M}\sum_{i=1}^{M}{\max{\{R(S|L_{i}),c\}}}.
\end{equation}
Note that, by Lemmas \ref{lemma:additive_submod} and \ref{lemma:max_supermod}, $F_{c}(S)$ is a supermodular function of $S$.


 An algorithm for approximating $S^{\ast}$ is as follows.    First, select parameters $\beta \geq 1$ and $\delta > 0$.  The algorithm finds a set $S$ satisfying $R(S) \leq R_{worst}^{\ast}$ and $|S| \leq \beta k$. Parameter $\delta$ determines the convergence speed of the algorithm.  Define $\alpha_{min}^{0} = 0$ and $\alpha_{max}^{0} = R_{max}$.  At the $j$-th iteration, let $\alpha^{j} = \frac{\alpha_{max}^{j-1} + \alpha_{min}^{j-1}}{2}$.  The goal of the $j$-th iteration is to determine if there is a set $S^{j}$ such that $|S^{j}| \leq \beta k$ and $R(S
 ^{j}|L_{i}) \leq \alpha^{j}$ for all $i=1,\ldots, M$.  This is accomplished by solving the optimization problem
 \begin{equation}
 \label{eq:worst_case_helper}
 \begin{array}{cc}
 \mbox{minimize} & |S| \\
 \mbox{s.t.} & F_{\alpha^{j}}(S) \leq \alpha^{j}
 \end{array}
 \end{equation}
 Since $F_{\alpha^{j}}(S)$ is supermodular, the solution to (\ref{eq:worst_case_helper}) can be approximated by an algorithm analogous to \textbf{static}-$\alpha$.  If the approximate solution to (\ref{eq:worst_case_helper}), denoted $S^{j}$, satisfies $|S^{j}| \leq \beta k$, then set $\alpha_{max}^{j} = \alpha^{j-1}$ and $\alpha_{min}^{j} = \alpha_{min}^{j-1}$.  Otherwise, set $\alpha_{max}^{j} = \alpha_{max}^{j-1}$ and $\alpha_{min}^{j} = \alpha^{j}$.  The algorithm terminates when $|\alpha_{max}^{j} - \alpha_{min}^{j}| < \delta$ and returns the current set $S^{j}$.  A pseudocode description of this algorithm is given as algorithm \textbf{switching}-$k$.


\begin{figure}[ht]
\centering
\renewcommand{\arraystretch}{0.6}
\begin{tabular}{lr}
\hline
\small{\textbf{Algorithm switching-$k$}: Algorithm for selecting up to $k$ leaders} & \\
 \small{to minimize worst-case error under switching topologies} & \\
 \hline
\small{\textbf{Input: } Topologies $G_{1}, \ldots, G_{M}$} & \\
~~~~~~~\small{Link error variances $\nu_{ij}^{(1)}, \ldots, \nu_{ij}^{(M)}$} & \\
~~~~~~~\small{Maximum number of leaders $\beta k$, threshold $\delta$} & \\
\small{\textbf{Output: } Set of leader nodes $S^{\ast}$} & \\
\small{\textbf{Initialization: } $S^{\ast} \leftarrow \emptyset$, $j \leftarrow 0$, $\alpha_{min}^{j} \leftarrow 0$, $\alpha_{max}^{j} \leftarrow R_{max}$} & \\
\small{\textbf{while} $\alpha_{max}^{j} - \alpha_{min}^{j} \geq \delta$} & \small{1} \\
~~\small{$\alpha^{j} \leftarrow \frac{\alpha_{max}^{j}+\alpha_{min}^{j}}{2}$} & \small{2} \\
~~\small{$r \leftarrow 0$, $S^{j} \leftarrow \emptyset$} & \small{3} \\
~~\small{\textbf{while} $F_{\alpha^{j}}(S) \leq \alpha^{j}$} & \small{4}\\
~~~~\small{$s_{i}^{\ast} \leftarrow \arg\max_{v \in V \setminus S}{\{F_{\alpha^{j}}(S) -  F_{\alpha^{j}}(S \cup \{v\})\}}$} & \small{5} \\
~~~~\small{$S^{j} \leftarrow S^{j} \cup \{s_{i}^{\ast}\}$, $r \leftarrow r+1$} & \small{6}\\
~~\small{\textbf{end while}} & \small{7} \\
~~\small{\textbf{if} $r > \beta k$} & \small{8} \\
~~~~\small{$\alpha_{max}^{j} \leftarrow \alpha^{j}$, $\alpha_{min}^{j} \leftarrow \alpha_{min}^{j-1}$} & \small{9} \\
~~\small{\textbf{else}} & \small{10} \\
~~~~\small{$\alpha_{min}^{j} \leftarrow \alpha^{j}$, $\alpha_{max}^{j} \leftarrow \alpha_{max}^{j-1}$} & \small{11} \\
~~\small{$j \leftarrow j+1$} & \small{12} \\
\small{\textbf{end while}} & \small{13} \\
\small{$S^{\ast} \leftarrow S^{j}$, \textbf{Return} $S^{\ast}$} & \small{14} \\
\hline
\end{tabular}
\label{algorithm:primal_switching}
\end{figure}

Since $\alpha_{max}^{j} - \alpha_{min}^{j}$ is strictly decreasing as $j$ increases, \textbf{switching-$k$} converges.  The optimality of \textbf{switching-$k$} is given by the following theorem.
\begin{theorem}
\label{theorem:switching_optimal}
When $\delta = \frac{1}{M}$ and $\beta$ satisfies
\begin{displaymath}
\beta \geq 1 + \log{\left(\max_{v \in V}{\left\{\sum_{i}{R(\{v\}|L_{i})}\right\}}\right)}
\end{displaymath}
\textbf{switching-$k$} returns a set $S^{\ast}$ such that $\max_{i}{\{R(S^{\ast}|L_{i})\}} \leq R_{worst}^{\ast}$ and $|S^{\ast}| \leq \beta k$.
\end{theorem}
\begin{proof}
The proof follows from the fact that $R(S|L_{i})$ is a supermodular function for all $i$ and Theorem 3 of \cite{Krause_adversarial}.
\end{proof}

\subsubsection{Choosing leaders to achieve an error bound}
\label{subsubsec:switching_error_bound}

In order to choose a minimum-size set of leaders to achieve an error bound $\alpha$, the following optimization problem must be solved

\begin{equation}
\label{eq:dual_switching_opt}
\begin{array}{cc}
\mbox{minimize} & |S| \\
\mbox{s.t.} & R(S|L_{i}) \leq \alpha \quad \forall i=1,\ldots, M \\
\end{array}
\end{equation}

The following lemma leads to efficient algorithms for solving (\ref{eq:dual_switching_opt}).

\begin{lemma}
Problem (\ref{eq:dual_switching_opt}) is equivalent to
\begin{equation}
\label{eq:dual_switching_opt_alternate}
\begin{array}{cc}
\mbox{minimize} & |S| \\
\mbox{s.t.} & F_{\alpha}(S) \leq \alpha
\end{array}
\end{equation}
where $F_{\alpha}(S)$ is defined as in (\ref{eq:def_F}).
\end{lemma}

\begin{proof}
The proof follows from the facts that the objective functions of (\ref{eq:dual_switching_opt}) and (\ref{eq:dual_switching_opt_alternate}) are the same, as well as the fact that $F_{\alpha}(S) \leq \alpha$ if and only if $R(S|L_{i}) \leq \alpha$ for all $i$.
\end{proof}

Since $F_{\alpha}(S)$ is a supermodular function of $S$, this is a supermodular optimization problem similar to (\ref{eq:dual_opt}), and hence can be solved by an algorithm analogous to \textbf{static-$\alpha$}.

\subsubsection{Leader selection for switching topologies under random agent and link failures}
\label{subsubsec:switching_failures}
As in the static network case, MAS with switching topologies may experience random link or agent failures.  In this case, the expected values of the average and worst-case system error are of interest when selecting the leaders.

Under random failures, the $i$-th topology can be represented as a random variable $\mathbf{G}_{i}$.  Let $\pi$ denote the joint distribution of $\mathbf{G}_{1}, \ldots, \mathbf{G}_{M}$, so that $\pi(G_{1}, \ldots, G_{M}) = Pr(\mathbf{G}_{1} = G_{1}, \ldots, \mathbf{G}_{M} = G_{M})$, and let $\pi_{i}(G) = Pr(\mathbf{G}_{i} = G)$.  Note that the $\mathbf{G}_{i}$'s may not be independent; for example, under the random failure model of Section \ref{subsec:random_topology}, the failure of an agent in any topology implies failure in all topologies.

The average-case expected error $\mathbf{E}_{\pi}(R_{avg}(S))$ can be further simplified by
\begin{displaymath}
\mathbf{E}_{\pi}(R_{avg}) = \mathbf{E}_{\pi}\left(\frac{1}{M}\sum_{i=1}^{M}{R(S|\mathbf{L}_{i})}\right)
 = \frac{1}{M}\sum_{i=1}^{M}{\mathbf{E}_{\pi_{i}}(R(S|L_{i}))},
 \end{displaymath}
 which follows from linearity of expectation and the fact that, by definition of $\pi_{i}$, $\mathbf{E}_{\pi}(R(S|L_{i})) = \mathbf{E}_{\pi_{i}}(R(S|L_{i}))$.  By Lemma \ref{lemma:random}, $\mathbf{E}_{\pi_{i}}(R(S|\mathbf{L}_{i}))$ is supermodular as a function of $S$ for all $i$.  Hence, $\mathbf{E}_{\pi}(R_{avg}(S))$ is a nonnegative weighted sum of supermodular functions, and is therefore a supermodular function of $S$ by Lemma \ref{lemma:additive_submod}.

 The problem of minimizing $\mathbf{E}_{\pi}(R_{avg}(S))$ when the number of leaders cannot exceed $k$ can be solved using an algorithm analogous to \textbf{static-$k$}.  Similarly, the problem of finding the smallest leader set $S$ that is within an upper bound $\alpha$ on $\mathbf{E}_{\pi}(R_{avg}(S))$ can be solved using an algorithm analogous to \textbf{static-$\alpha$}.  Both modified algorithms take $\pi_{1}, \ldots, \pi_{M}$ as additional inputs and have line 3 changed to
 \begin{displaymath}
 s_{i}^{\ast} \leftarrow \arg\max_{v \in V \setminus S}{\left\{\frac{1}{M}\sum_{i=1}^{M}{\mathbf{E}_{\pi_{i}}(R(S|\mathbf{L}_{i}) - R(S \cup \{v\}|\mathbf{L}_{i}))}\right\}}.
 \end{displaymath}

 Considering the expected worst-case error, $\mathbf{E}_{\pi}(R_{worst}(S))$, observe that, by Jensen's inequality~\cite{probability},
 \begin{equation}
  \label{eq:switching_random_lower_bound}
 \mathbf{E}_{\pi}\left(\max_{i=1,\ldots, M}{\{R(S|\mathbf{L}_{i})\}}\right) \geq \max_{i=1,\ldots, M}{\{\mathbf{E}_{\pi}(R(S|\mathbf{L}_{i}))\}} 
  = \max_{i=1,\ldots, M}{\left\{\mathbf{E}_{\pi_{i}}(R(S|\mathbf{L}_{i}))\right\}}.
  \end{equation}
  The lower bound  (\ref{eq:switching_random_lower_bound}) is the worst-case expected error experienced when the leader set is $S$.    Since $\mathbf{E}_{\pi_{i}}(R(S|\mathbf{L}_{i}))$ is supermodular as a function of $S$, the function
  \begin{equation}
  \label{eq:random_failures_switching_helper}
  \tilde{F}_{c}(S) = \max{\left\{\mathbf{E}_{\pi_{i}}(R(S|\mathbf{L}_{i})), c\right\}}
  \end{equation}
  is supermodular in $S$ by Lemma \ref{lemma:max_supermod}.  Therefore, the problem of selecting up to $k$ leaders in order to minimize $\mathbf{E}_{\pi}(R_{worst}(S))$ can be approximately solved by an algorithm similar to \textbf{switching-$k$}.  The modified algorithm takes $\pi_{1}, \ldots, \pi_{M}$ as additional input, and replaces the function $F_{\alpha^{j}}(S)$ at line 5 with the function $\tilde{F}_{\alpha^{j}}(S)$ defined in (\ref{eq:random_failures_switching_helper}).


  In order to choose the minimum-size set of leaders such that $\max_{i=1,\ldots, M}{\left\{\mathbf{E}_{\pi}(R(S|\mathbf{L}_{i}))\right\}}$ is below an error bound $\alpha$, a supermodular optimization problem analogous to (\ref{eq:dual_switching_opt}) can be used.   The constraint  of the modified problem is given by $\tilde{F}_{\alpha}(S) \leq \alpha$ with $\tilde{F}_{\alpha}(S)$ defined as in (\ref{eq:random_failures_switching_helper}).

\subsection{Leader Selection Under Arbitrarily Time-Varying Topologies}
\label{subsec:arbitrary}

In this section, MAS with topologies that vary arbitrarily in time are considered.  We assume that time is divided into steps, and that at step $t$, the MAS owner has knowledge of the topologies $G_{1}, \ldots, G_{t-1}$ for steps $1, 2, \ldots, t-1$, respectively, as well as their corresponding weighted Laplacians $L_{1}, \ldots, L_{t-1}$, but does not know the topology $G_{t}$ for step $t$.  Furthermore, we assume that, while the topology may vary arbitrarily over time, the variation is sufficiently slow that the agent dynamics approximate the steady-state for graph topology $G_{t}$ during the $t$-th time step.

Under this model, a fixed set of leaders chosen at step $t=1$, may give poor performance for the subsequent topologies $G_{2}, \ldots, G_{T}$.  Instead, it is assumed that, at step $t$, a new set of leaders $S_{t}$ is selected based on the observed topologies $G_{1}, \ldots, G_{t-1}$.  The error for each topology is given by $R(S_{t}|L_{t})$. The leader selection problem for time-varying topologies is stated as
\begin{equation}
\label{eq:arbitrary_opt}
\begin{array}{cc}
\mbox{minimize}_{S_{1}, \ldots, S_{T}} & \sum_{t=1}^{T}{R(S_{t}|L_{t})} \\
\mbox{s.t.} & |S_{t}| \leq k
\end{array}
\end{equation}

Problem (\ref{eq:arbitrary_opt}) is an online supermodular optimization problem~\cite{streeter2007online}.
The method for choosing a set of leader nodes $S_{t}^{\ast}$ for the $t$-th time step is as follows.  Consider the \textbf{static}-$k$ algorithm.  Since $G_{t}$ is unknown and random,
the node $j$ that maximizes $\{R(S_{t,i}^{\ast}) - R(S_{t,i}^{\ast} \cup \{j\})\}$ at the $i$-th iteration of the algorithm is also random.  Let
\begin{equation}
\label{eq:time_varying_distribution}
\pi_{t,i}(l) = Pr(\arg\max_{j}{\{R(S_{t,i}^{\ast}) - R(S_{t,i}^{\ast} \cup \{j\})\}} = l).
\end{equation}
Then for step $t$, instead of selecting $s_{t,i}^{\ast}$ deterministically as in Line 3 of \textbf{static}-$k$, $s_{t,i}^{\ast}$ is selected probabilistically with distribution $\pi_{t,i}$ in (\ref{eq:time_varying_distribution}).

In general, the exact values of $\pi_{t,i}$ will not be known during leader selection.  To address this, an online learning technique  is used to estimate $\pi_{t,i}$ based on observations from the previous $t-1$ time steps.  Under this approach, a set of weights $\mathbf{w}_{t,1}, \ldots, \mathbf{w}_{t,k}$ is maintained, where $\mathbf{w}_{t,i}$ is a vector in $\mathbf{R}^{n}$ with $\mathbf{w}_{t,i}(j)$ representing the weight assigned to choosing node $j$ as the $i$-th leader during step $t$.  Define
\begin{equation}
\label{eq:s_opt}
s_{t,i}^{opt} \triangleq \arg\max_{j}{\{R(S^{\ast}_{t,i-1}|L_{t}) - R(S^{\ast}_{t,i-1} \cup \{j\}|L_{t})\}}.
\end{equation}
In other words $s_{t,i}^{opt}$ in (\ref{eq:s_opt}) is the best possible choice of $s_{t,i}^{\ast}$ for the topology $G_{t}$.
Then define the loss $l_{t,i,j}$ associated with choosing $s_{t,i}^{\ast}=j$ to be
\begin{equation}
l_{t,i,j} \triangleq 1 - \frac{R(S_{t}^{\ast}|L_{t}) - R(S_{t}^{\ast} \cup \{j\}|L_{t})}{R(S_{t}^{\ast}|L_{t}) - R(S_{t}^{\ast} \cup\{s_{t,i}^{opt}\}|L_{t})}.
\end{equation}
At the end of step $t$, the value of $\mathbf{w}_{t,i}(j)$ is updated to
$\mathbf{w}_{t+1,i}(j) = \beta^{l_{t,i,j}}\mathbf{w}_{t,i}(j)$,
where $\beta \in (0,1]$ is a system parameter that can be tuned to adjust the performance of the learning algorithm.  This is interpreted as penalizing  node $j$ for suboptimal performance during interval $t$.  By decreasing the value of $\beta$, nodes experiencing higher losses $l_{t,i,j}$ will be much less likely to be selected during the $(t+1)$-th time step.
$\mathbf{w}_{t+1,i}$ is then normalized to obtain the estimated distribution $\pi_{t+1,i}$.  This process is described in detail in algorithm \textbf{online}-$k$.

In analyzing this approach, the total error $\sum_{t}{R(S_{t}|G_{t})}$ can be compared to the error achievable when all $T$ topologies are known in advance.  The following theorem gives a bound on the difference between these two errors.
\begin{figure}
\centering
\renewcommand{\arraystretch}{0.6}
\begin{tabular}{lr}
\hline
\small{\textbf{Algorithm online-$k$}: Algorithm for selecting up to $k$ leaders } & \\
\small{for an arbitrarily time-varying topology} & \\
\hline
\small{\textbf{Input:} Current weights $\mathbf{w}_{t,1}, \ldots, \mathbf{w}_{t,k}$} & \\
~~~~~~~\small{$G_{t} = (V,E_{t})$, link error variances $\nu_{ij}^{t}$} & \\
~~~~~~~\small{Parameter $\beta \in (0,1]$} & \\
~~~~~~~\small{Maximum number of leader nodes $k$} & \\
~~~~~~~\small{Current set of leader nodes $S^{\ast}_{t}$} & \\
\small{\textbf{Output:} Updated weights $\mathbf{w}_{t+1,1}, \ldots, \mathbf{w}_{t+1,k}$} & \\
~~~~~~~~~\small{Set  of leader nodes $S_{t+1}^{\ast}$} & \\
\small{\textbf{Initialization:} $S_{t+1}^{\ast} \leftarrow \emptyset$} & \small{1}\\
\small{\textbf{for} $i=1,\ldots,k$} & \small{2}\\
\small{~~~~\small{$s_{t,i}^{opt} \leftarrow \arg\max_{j}{\{R(S_{t,i-1}^{\ast}) - R(S_{t,i-1}^{\ast} \cup \{j\})\}}$}} & \small{3} \\
~~~~\small{\textbf{for} $j=1,\ldots,n$} & \small{4}\\
~~~~~~~~\small{$l_{t,i,j} \leftarrow 1 - \frac{R(S_{t,i-1}^{\ast}) - R(S_{t,i-1}^{\ast} \cup \{j\})}{R(S_{t,i-1}^{\ast}) - R(S_{t,i-1}^{\ast} \cup \{s_{t,i}^{opt}\})}$} & \small{5} \\
~~~~~~~~\small{$\mathbf{w}_{t+1,i}(j) \leftarrow (\mathbf{w}_{t,i}(j))\beta^{l_{t,i,j}}$} & \small{6}\\
~~~~\small{\textbf{end}} & \small{7}\\
~~~~\small{$\mathbf{\pi}_{t+1,i} \leftarrow \mathbf{w}_{t+1,i}/\mathbf{1}^{T}\mathbf{w}_{t+1,i}$} & \small{8}\\
~~~~\small{Choose $s_{t+1,i}^{\ast}$ randomly  with distribution $\pi_{t+1,i}$} & \small{9}\\
~~~~\small{$S_{t+1}^{\ast} \leftarrow S_{t+1}^{\ast} \cup \{s_{t+1,i}^{\ast}\}$} & \small{10}\\
\small{\textbf{end}} & \small{11}\\
\hline
\end{tabular}
\label{algorithm:iteration}
\vspace{-20pt}
\end{figure}
\begin{theorem}
\label{theorem:optimality_online}
Suppose that the algorithm \textbf{online-$k$} is executed for $T$ steps, and let $G_{1}, \ldots, G_{T}$ be the topologies during those steps.  Let $R_{max}$ be defined as in Theorem \ref{theorem:primal_static_opt_bound}.
Define the error $K$ to be
\begin{equation}
K \triangleq (1-1/e)\sum_{t=1}^{T}{R(S_{t}|L_{t})} - \left(\max_{|S|=k}{\left\{\sum_{t=1}^{T}{R(S|L_{t})}\right\}}\right).
\end{equation}
Then $K \leq \mathcal{O}(\sqrt{R_{max}kT\log{n}})$.
\end{theorem}
\begin{proof}
$R(S|L_{1}), \ldots, R(S|L_{T})$ is a sequence of supermodular functions bounded above by $R_{max}$.  Then, by Lemma 4 of \cite{streeter2007online},
\begin{equation}
\label{eq:online_proof_helper}
(1-1/e)\sum_{t=1}^{T}{R((S_{t}^{\ast}|L_{t}}) - \left(\max_{|S|=k}{\left\{\sum_{t=1}^{T}{R(S|L_{t})}\right\}}\right)
\leq \mathcal{O}(\sqrt{R_{max}kT\log{n}})
\end{equation}
as desired.
\end{proof} 
\section{Simulation Study}
\label{sec:simulation}

\noindent In this section, we evaluate the performance of our leader selection algorithms.  
\label{parameters}
Simulations are carried out using Matlab.  A network of 100 agents is simulated, with agents placed at random positions within a $1000$m x $1000$m rectangular area. Two agents are assumed to share a link if they are deployed within $300$m of each other.  The error variance $\nu_{ij}$ of link $(i,j)$ is assumed to be proportional to the distance between agents $i$ and $j$.  Each data point in the following figures represents an ensemble average of $50$ trials, unless otherwise indicated.

For comparison, five different leader selection algorithms are simulated.  In the first algorithm, a random subset of agents is chosen to act as leaders.  In the second algorithm, the $k$ nodes with highest degree (i.e., largest number of neighbors) are chosen to act as leaders.
In the third algorithm, the $k$ nodes with degree closest to the average degree are selected.  The fourth algorithm simulated, for the case of selecting up to $k$ leaders in a static topology, is the convex optimization approach of \cite{fardad2011noisefree}.
The fifth algorithm  was our supermodular optimization.

\begin{figure}[h]
\centering
$\begin{array}{cc}
\includegraphics[width=3in]{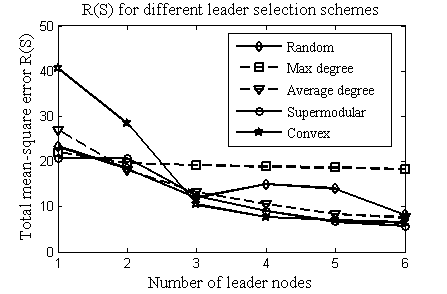} &
\includegraphics[width=3in]{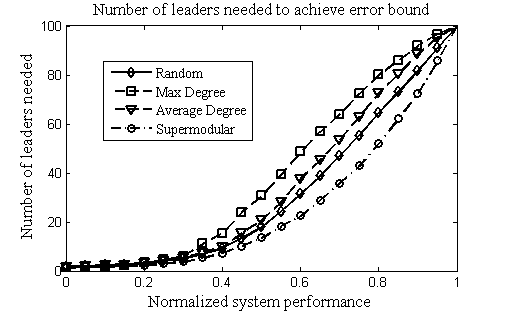} \\
(a) & (b) \\
\end{array}$
\caption{(a) A comparison of our supermodular optimization approach to leader selection with the convex optimization approach of \cite{fardad2011noisefree}, as well as random and degree-based leader selection.  Either the supermodular or convex optimization approach provides minimal error, depending on the number of leader agents. (b) The supermodular optimization approach \textbf{static-$\alpha$} requires fewer leaders to satisfy the error constraint than either the degree-based or random heuristics.  Error values are normalized to between 0 and 1.}
\label{fig:static_comparison}
\end{figure}

\underline{\emph{Case 1: MAS with static network topology -- }} Figure \ref{fig:static_comparison}(a) compares the performance of the five algorithms considered for the
problem of choosing up to $k$ leaders in order to minimize the total system error.
For this comparison, in order to reduce the runtime of the  convex optimization
approach, a smaller network of 25 nodes is used.   Figure \ref{fig:static_comparison}(a) shows the error achieved by the different leader selection algorithms for a fixed network topology and varying leader set size, in which either the convex or supermodular approaches provide optimal performance depending on the number of leader agents.  When $k=1,2,5,6$, our supermodular optimization approach results in lower mean-square error, while the convex optimization approach in \cite{fardad2011noisefree} selects leaders with lower error when $k=3,4$.

For the problem of choosing the minimum number of leaders to achieve an error bound, the supermodular optimization approach requires only 40 leaders to achieve normalized error of $0.7$ (for example), compared to $50$ leaders for the random heuristic and over 60 leaders for the maximum degree method (Figure \ref{fig:static_comparison}(b)).
Figure \ref{fig:static_comparison}(b) also suggests that the random heuristic consistently outperforms both degree-based algorithms.  Selecting the nodes with average degree also performs better than selecting the maximum-degree nodes to act as leaders.

\begin{figure}[h]
\centering
$\begin{array}{cc}
\includegraphics[width=3in]{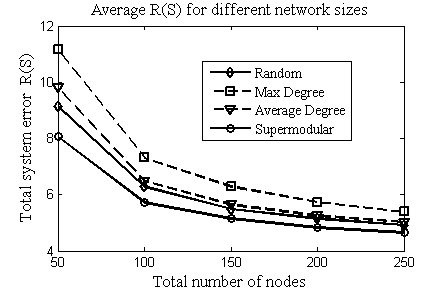} &
\includegraphics[width=3in]{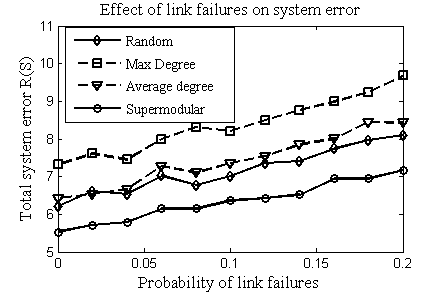} \\
(a) & (b)
\end{array}$
\caption{(a) Error experienced by MAS as the number of agents increases.  For each network size, the supermodular optimization approach provides the lowest overall error, followed by random selection, average degree-based selection, and maximum degree-based selection.  Overall error decreases in all cases due to the increase in node density.  (b) Effect of random failures on MAS error when number of agents is $100$ and number of leaders is $10$.  Although all selection schemes experience an increase in error due to link failures, the increase is smallest for the supermodular approach.}
\label{fig:network_size}
\end{figure}

The total system error experienced by the network as a function of network size is explored in Figure \ref{fig:network_size}(a).  The number of leaders is equal to $0.1n$, where $n$ is the number of agents.  Since the deployment area remains constant, adding agents to the network increases the number of links, resulting in smaller overall error.  Hence, while the supermodular approach still outperforms the other methods, the difference in overall error decreases as the node density grows large.

~\\

\underline{\emph{Case 2: MAS experiencing random link failures -- }}
Figure \ref{fig:network_size}(b) shows the error experienced for each method when links fail independently and at random, with probability ranging from $0$ to $0.2$.  The number of leaders is equal to $10$ for each scheme, while the network size is $100$.  The supermodular optimization algorithm uses the Monte Carlo approach described in Section \ref{subsec:random_topology}.  While each scheme sees a degradation in performance as the probability of failure increases, this degradation is minimized by the supermodular optimization method.



\begin{figure}[h]
\centering
$\begin{array}{cc}
\includegraphics[width=3in]{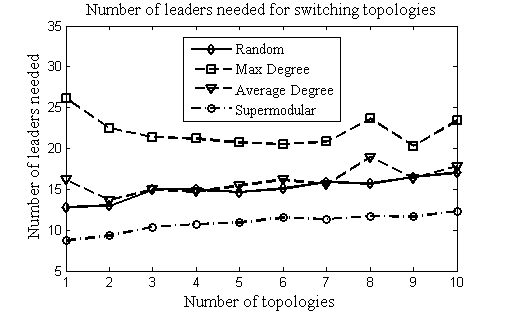} &
\includegraphics[width=3in]{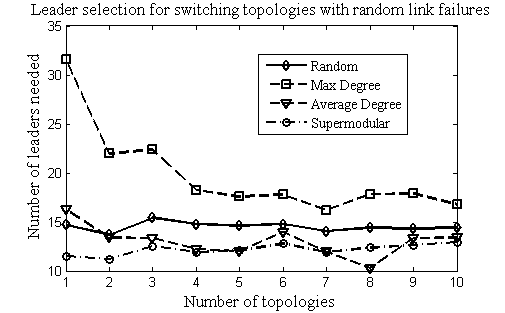} \\
(a) & (b)
\end{array}$
\caption{(a) Number of leaders required to achieve system error of 6 when the MAS switches between $M$ randomly generated topologies, where $M$ varies from $1$ to $10$.  The number of leaders needed increases with the number of topologies; the supermodular selection method requires the fewest number of leaders. (b) Number of leaders required when the MAS switches between $M$ randomly generated topologies and links fail independently with probability $p=0.05$.  The number of leaders required is greater due to link failures.}
\label{fig:switching}
\end{figure}

\underline{\emph{Case 3: MAS that switch between predefined topologies --}}
Figure \ref{fig:switching}(a) shows the number of leaders needed to achieve an error level of $6$ for each algorithm for MAS under switching topologies.  For this evaluation, a set of $M$ topologies, where $M$ varied from $1$ to $10$, is generated at random based on the deployment area and node communication range described in the first paragraph of Section \ref{sec:simulation}.  A fixed leader set $S$ is then selected using each heuristic.  Each scheme requires a larger leader set as the number of prespecified topologies increased; however, for the supermodular optimization approach, a fixed set of $10$ leaders provides an error of less than $6$ for 10 different topologies.
Overall, fewer leaders are needed for the supermodular optimization approach.  Random leader selection requires fewer leaders than the average degree heuristic, which in turn outperforms selection of the highest-degree nodes as leaders.

The case of MAS with switching topologies and independent, random link failures is shown in Figure \ref{fig:switching}(b).  Link failures increase the number of leaders required to achieve the error bound for each value of the number of topologies, $M$. While the supermodular optimization approach continued to perform better than the other heuristics in most cases, the performance improvement was less significant.  

\begin{figure}[h]
\centering
\includegraphics[width=3in]{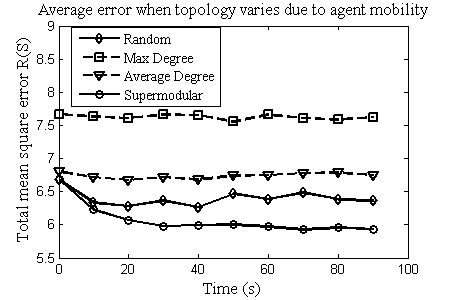}
\caption{Leader selection under arbitrary time-varying topologies.  Nodes move according to a group mobility model~\cite{hong1999group} with speed 30 m/s.  A new set of leaders is selected every 10 seconds.  While the performance of the four selection algorithms is comparable, the online supermodular approach \textbf{online-$k$} performs better over time by incorporating observed network topology information.}
\label{fig:arbitrary}
\end{figure}

\underline{\emph{Case 4: MAS with arbitrarily time-varying topology -- }}
Figure \ref{fig:arbitrary} shows the performance of leader selection schemes when the topology varies over time due to agent mobility.  Nodes are assumed to move according to a group mobility model, in which nodes attempt to maintain their positions with respect to a reference point~\cite{hong1999group}.  The reference point varies according to a random walk with speed 30 m/s.  Each node's position is equal to its specified position relative to the reference plus a uniformly distributed error.  A new set of $10$ leaders is selected every 10 seconds.
As in the other cases, the supermodular optimization approach consistently provides the lowest error, followed by the random, average degree, and maximum degree heuristics.  Moreover, the \textbf{online-$k$} algorithm improves its performance over time by observing which agents provided the best performance when chosen as leaders and assigning those agents a higher weight.

\section{Conclusions}
\label{sec:conclusion}
\noindent In this paper, the problem of selecting leaders in linear multi-agent systems in order to minimize error due to communication link noise was studied.  We analyzed the total mean-square error in the follower agent states, and formulated the problem of selecting up to $k$ leaders in order to minimize the error, as well as the problem of selecting the minimum-size set of leaders to achieve a given upper bound on the error.  We examined both problems for different cases of MAS, including MAS with (a) static network topology, (b) topologies that experience random link  failures, (c) switching between predefined topologies, and (d) topologies that vary arbitrarily over time.  We showed that all of these cases can be solved within a supermodular optimization framework.  We introduced efficient algorithms for selecting a set of leaders that approximates the optimum set up to a provable bound for each of the four cases.   
 Our proposed approach was evaluated and compared with other methods, including random leader selection, selecting high-degree agents as leaders, selecting average-degree agents as leaders, and a convex optimization-based approach through a simulation study.  Our study showed that supermodular leader selection significantly outperformed the random and degree-based leader selection algorithms in static as well as dynamic MAS while providing provable bounds on the MAS performance, and provides performance comparable to the convex optimization approach.



\bibliographystyle{IEEEtran}
\bibliography{Clark_TAC_2011}
\appendices
\section{Derivation of Agent Dynamics}
\label{sec:derivation_agent_dynamics}
In this appendix, the dynamics of each agent are derived following the analysis in \cite{barooah2006graph}, and a proof that the error variance is equal to $\frac{1}{2}\mathbf{tr}(L_{ff}^{-1})$ is provided.  The goal of each follower agent $i$ is to estimate $x_{i}-x_{i}^{\ast}$, the deviation from the desired state.  Since $x_{i}^{\ast} - x_{j}^{\ast} = r_{ij}$, we assume that $x_{i}^{\ast} = x_{j} + r_{ij}$.  Each agent $i$'s relative state estimate for each $j \in N(i)$ can therefore be written as
\begin{displaymath}
y_{ij} = x_{i} - x_{j} + \epsilon_{ij} = x_{i}-x_{i}^{\ast} + r_{ij} + \epsilon_{ij}.
\end{displaymath}
Letting $\hat{y}_{ij} = y_{ij}-r_{ij} = x_{i}-x_{i}^{\ast} + \epsilon_{ij}$, the goal of agent $i$ is to estimate $x_{i}-x_{i}^{\ast}$ given the system of equations
\begin{displaymath}
\hat{y}_{ij} = x_{i}-x_{i}^{\ast} + \epsilon_{ij}, \quad j \in N(i),
\end{displaymath}
which can be written in vector form as $\hat{\mathbf{y}} = \mathbf{1}(x_{i}-x_{i}^{\ast}) + \epsilon$. By the Gauss-Markov Theorem, the best linear unbiased estimator of  $x_{i}-x_{i}^{\ast}$ is equivalent to the least squares estimator, which is given by $(\mathbf{1}^{T}C\mathbf{1})^{-1}\mathbf{1}^{T}C^{-1}\mathbf{\hat{y}}$, where $C$ is the diagonal covariance matrix of $\epsilon$.  The estimator reduces to
\begin{displaymath}
x_{i}-x_{i}^{\ast} \approx \left(\sum_{j \in N(i)}{\nu_{ij}^{-1}}\right)^{-1}\sum_{j \in N(i)}{\nu_{ij}^{-1}\hat{y}_{ij}},
\end{displaymath}
thus motivating the choice of link weights $D_{i}^{-1}\nu_{ij}^{-1}$, where $D_{i} = \sum_{j \in N(i)}{\nu_{ij}^{-1}}$.  The following lemma, which first appeared in \cite{barooah2006graph}, characterizes the asymptotic error variance of the follower agent states under these dynamics.

\begin{lemma}
\label{lemma:appendix_error_proof}
Let $X(t)$ denote the covariance matrix of $\mathbf{x}_{f}(t)$.  Then $\lim_{t \rightarrow \infty}{X(t)} = \frac{1}{2}L_{ff}^{-1}$.
\end{lemma}

\begin{proof}
The dynamics of $\mathbf{x}_{f}(t)$ are given by
\begin{displaymath}
\dot{\mathbf{x}}_{f}(t) = -D_{f}^{-1}(L_{ff}\mathbf{x}_{f}(t) + L_{fl}\mathbf{x}_{l}^{\ast} + B\mathbf{r}) + \mathbf{w}(t),
\end{displaymath}
where $\mathbf{w}(t)$ is a zero-mean white process with variance $D_{f}^{-1}$.  Asymptotically, the covariance matrix of $\mathbf{x}_{f}(t)$ is given as the positive definite solution $X$ to the Lyapunov equation
\begin{displaymath}
-D_{f}^{-1}L_{ff}X - XL_{ff}D_{f}^{-1} + D_{f}^{-1} = 0.
\end{displaymath}
By inspection, $X = \frac{1}{2}L_{ff}^{-1}$.
\end{proof}

\section{Proof of Lemma \ref{lemma:max_supermod}}
\label{sec:supermod_lemma_proof}
Lemma \ref{lemma:max_supermod} is restated and proved as follows.
\begin{lemma}
Let $f:2^{V} \rightarrow \mathbf{R}$ be a nonincreasing supermodular function, and let $c \geq 0$ be constant.  Then $g(S) = \max{\{f(S),c\}}$ is supermodular.
\end{lemma}
\begin{proof}
It is enough to show that for any $c$, $F(S) \triangleq \max{\{f(S), c\}}$ is supermodular.  The proof uses the fact that a function $F$ is supermodular if and only if (\cite{fujishige2005submodular})
\begin{equation}
\label{eq:equiv_submod_def}
F(A) + F(B) \leq F(A \cup B) + F(A \cap B)
\end{equation}
as well as the fact that $f(A \cup B) \leq (f(A), f(B)) \leq f(A \cap B)$.  There are four cases.

\noindent \underline{\emph{Case 1:} $\alpha < f(A \cup B)$:} In this case, (\ref{eq:equiv_submod_def}) follows from the supermodularity of $f$.

\noindent \underline{\emph{Case 2:} $f(A) < c$, $f(B) > c$:} Under this case, (\ref{eq:equiv_submod_def}) is equivalent to $f(B) \leq f(A \cap B)$, which follows from the monotonicity property.  The case where $f(A) > c$ and $f(B) < c$ is similar.

\noindent \underline{\emph{Case 3:} $f(A) < c$, $f(B) < c$, $f(A \cap B) > c$:} For this case, (\ref{eq:equiv_submod_def}) is equivalent to $f(A \cap B) \geq \alpha$, which is true by assumption.

\noindent \underline{\emph{Case 4:} $f(A \cap B) \leq \alpha$:} Eq. (\ref{eq:equiv_submod_def}) is trivially satisfied in this case.
\end{proof}

\section{Proof of Theorem \ref{theorem:commute_time}}
\label{sec:appendix}
In this appendix, a proof of Theorem \ref{theorem:commute_time} is given.
 Before proving Theorem \ref{theorem:commute_time}, the following intermediate lemmas are needed.

\begin{lemma}
\label{lemma:intermediate0}
Consider the equation $Lv = J$.  Define $v^{\ast}$ to be the unique vector satisfying $Lv^{\ast} = J$.  When  $v_{u}^{\ast} = 1$, $J_{i} = 0$ for $i \in V \setminus (S+u)$, and $v_{s}^{\ast} = 0$ for $s \in S$,   $(L_{f}^{-1})_{uu}$ is equal to $J_{u}^{-1}$.
\end{lemma}

\begin{proof}
Write $v^{\ast} = [v_{f}^{\ast T} \quad 0]^{T}$ and $J = [J_{f}^{T} \quad J_{l}^{T}]^{T}$.  The equation $Lv^{\ast} = J$ then reduces to $L_{f}v_{f}^{\ast} = J_{f}$, which is equivalent to $v_{f}^{\ast} = L_{f}^{-1}J_{f}$.  Multiplying both sides of the equation by $e_{u}^{T}$, where $e_{u}$ has a $1$ in the $u$-th entry and $0$s elsewhere, yields $v_{u}^{\ast} = e_{u}^{T}L_{f}^{-1}J_{f}$.  The right hand side can be expanded to
\begin{equation}
\label{eq:intermediate0}
(L_{f}^{-1})_{u1}J_{1} + \cdots + (L_{f}^{-1})_{u(n-|S|)}J_{n-|S|} = 1.
\end{equation}
By definition of $J$, all terms of the left-hand side of (\ref{eq:intermediate0}) are zero except $(L_{f}^{-1})_{uu}J_{u}$.  Dividing both sides by $J_{u}$ yields the desired result.
\end{proof}

\begin{lemma}
\label{lemma:intermediate1}
  Let $v^{\ast}$ be defined as in Lemma \ref{lemma:intermediate0}. Then $v^{\ast}$ satisfies
\begin{equation}
\label{eq:intermediate1_equation}
v_{i}^{\ast} = \sum_{j \in N(i)}{P(i,j)v_{j}^{\ast}} \qquad \forall i \in V \setminus (S+u),
\end{equation}
where $P(i,j)$ is defined as in (\ref{eq:probability_distribution}).
\end{lemma}

\begin{proof}
By definition of $L$ and the fact that $J_{i} = 0$,
\begin{equation}
\left(\sum_{j \in N(i)}{L(i,j)}\right)v_{i}^{\ast} = \sum_{j \in N(i)}{L(i,j)v_{j}^{\ast}}.
\end{equation}
Dividing both sides by $\sum_{j \in N(i)}{L(i,j)}$ gives (\ref{eq:intermediate1_equation}).
\end{proof}


\begin{lemma}
\label{lemma:intermediate2}
Define $\tilde{v}_{i}(S,u)$ to be the probability that a random walk starting at $i$ with transition probabilities given by (\ref{eq:probability_distribution}) reaches node $u \in V \setminus S$ before any node in the set $S$.  Then $\tilde{v}_{i} = v^{\ast}_{i}$ for all $i$.
\end{lemma}

\begin{proof}
By definition, $\tilde{v}_{i} = v^{\ast}_{i}$ for all $i \in S$ and for $i=u$.  It remains to show the result for $i \in V \setminus (S+u)$.  By the stationarity of the random walk,
\begin{equation}
\label{eq:harmonic}
\tilde{v}_{i} = \sum_{j \in N(i)}{P(i,j) \tilde{v}_{j}}.
\end{equation}
Hence, by (\ref{eq:harmonic}) and Lemma \ref{lemma:intermediate1}, both $v^{\ast}$ and $\tilde{v}$ are harmonic functions of the index $i$.  By the Maximum Principle \cite{lovasz1993random}, $v^{\ast} = \tilde{v}$.
\end{proof}


\begin{proof}[Proof of Theorem \ref{theorem:commute_time}]
By Lemma \ref{lemma:intermediate0},
\begin{equation}
\label{eq:Lf_inv}
(L_{f}^{-1})_{uu} = \left(\sum_{t \in N(u)}{\nu_{ut}^{-1}(1 - v_{t}^{\ast})}\right)^{-1}.
\end{equation}
Hence, it suffices to show that the commute time $\kappa (S,u)$ is proportional to the right-hand side of (\ref{eq:Lf_inv}).
To show this, first observe that the probability $Q(u,S)$ that a random walk originating at $u$ will reach $S$ before returning to $u$ is given by $1 - \sum_{t \in N(u)}{\tilde{v}_{t}P(u,t)}$, where $\tilde{v}$ is as in Lemma \ref{lemma:intermediate2}.
Proposition 2.3 of \cite{lovasz1993random} gives $Q(u,S)$ as a function of the commute time,
\begin{displaymath}
Q(u,S) = \frac{2 \sum_{(s,t) \in E}{\nu_{st}^{-1}}}{\kappa(S,u)D_{u}}.
\end{displaymath}
This yields
\begin{displaymath}
1 - \sum_{t \in N(u)}{\tilde{v}_{t}P(u,t)} = \frac{2 \sum_{(s,t) \in E}{\nu_{st}^{-1}}}{\kappa(S,u)D_{u}}.
\end{displaymath}
Rearranging terms and applying Lemma \ref{lemma:intermediate2} gives
\begin{IEEEeqnarray*}{rCl}
\kappa(S,u) &=& \left(2 \sum_{(s,t) \in E}{\nu_{st}^{-1}}\right)\left[D_{u}\left(1 - \sum_{t \in N(u)}{v^{\ast}_{t}P(u,t)}\right)\right]^{-1} \\
&=& \left(2 \sum_{(s,t) \in E}{\nu_{st}^{-1}}\right)\left[\sum_{t \in N(u)}{\nu_{uv}^{-1}(1 - v_{t}^{\ast})}\right]^{-1} = \left(2 \sum_{(s,t) \in E}{\nu_{st}^{-1}}\right)(L_{ff}^{-1})_{uu},
\end{IEEEeqnarray*}
which proves the theorem.
\end{proof}

\section{Alternate Proof of Theorem \ref{theorem:commute_time_supermodular}}
\label{sec:alternate_proof_main_theorem}
In this section, an alternative proof to Theorem \ref{theorem:commute_time_supermodular} is presented, which uses an illustrative, graphical approach instead of the derivation with indicator functions in Section \ref{sec:problem_static}.  We first restate Theorem \ref{theorem:commute_time_supermodular}.

\begin{theorem}
    The commute time $\kappa(S,u)$ is a nonincreasing supermodular function of $S$.
    \end{theorem}

    \begin{proof}
    The nonincreasing property follows from the fact that, if $S \subseteq T$, then any random walk that reaches $S$ and returns to $u$ has also reached $T$ and returned to $u$.  Thus $\kappa(T,u) \leq \kappa(S,u)$.

    \qquad By definition of submodularity~\cite{fujishige2005submodular}, $\kappa(S,u)$ is a supermodular function of $S$ if and only if, for any sets $A$ and $B$ with $A \subseteq B$ and any $j \notin B$,
    \begin{equation}
    \label{eq:supermod_def_Thm_4}
    \kappa(A,u) - \kappa(A \cup \{j\}, u) \geq \kappa(B,u) - \kappa(B \cup \{j\}, u).
    \end{equation}
    Consider the quantity $\kappa(A,u) - \kappa(A \cup \{j\},u)$.  Define $A^{\prime} = A \cup \{j\}$, and define $T_{Au}$ and $T_{A^{\prime}u}$ to be the (random) times for a random walk to reach $A$ (respectively $A^{\prime}$) and return to $u$.  Then by definition, $\kappa(A,u) = \mathbf{E}(T_{Au})$ and $\kappa(A^{\prime},u) = \mathbf{E}(T_{A^{\prime}u})$.  This implies $\kappa(A,u) - \kappa(A^{\prime},u) = \mathbf{E}(T_{Au} - T_{A^{\prime}u})$.

    \qquad Let $\tau_{j}(A)$ denote the event where the random walk reaches node $j$ before any of the nodes in $A$.  Further, let $h_{jAu}$ be the time for the random walk to travel from $j$ to $A$ and then to $u$, while $h_{ju}$ is the time to travel directly from $j$ to $u$.  We then have
    \begin{IEEEeqnarray*}{rCl}
    \kappa(A,u) - \kappa(A^{\prime},u) &=& \mathbf{E}(T_{Au} - T_{A^{\prime}u}|\tau_{j}(A))Pr(\tau_{j}(A))
      + \mathbf{E}(T_{Au} - T_{A^{\prime}u}|\tau_{j}(A)^{c})Pr(\tau_{j}(A)^{c}) \\
     &=& \mathbf{E}(h_{jAu} - h_{ju})Pr(\tau_{j}(A)),
     \end{IEEEeqnarray*}
     noting that, if the walk reaches $A$ first, then $T_{Au}$ and $T_{A^{\prime}u}$ are equal.  Hence (\ref{eq:supermod_def_Thm_4}) is equivalent to
     \begin{equation}
     \label{eq:submitted_13}
     \mathbf{E}(h_{jAu}-h_{ju})Pr(\tau_{j}(A)) \geq \mathbf{E}(h_{jBu} - h_{ju})Pr(\tau_{j}(B)).
     \end{equation}


\qquad To complete the proof, it suffices to show that $h_{jAu} \geq h_{jBu}$ and $Pr(\tau_{j}(A)) \geq Pr(\tau_{j}(B))$.  If the walk reaches $j$ before $B$, then it has also reached $j$ before $A$, since $A \subseteq B$.  Thus $\tau_{j}(B) \subseteq \tau_{j}(A)$, and hence $Pr(\tau_{j}(B)) \leq Pr(\tau_{j}(A))$.

\qquad It remains to show that $h_{jAu} \geq h_{jBu}$.  We consider three cases, corresponding to different sample paths of the random walk on $G$. As a preliminary, let $\zeta_{ab}$ denote the time for a random walk starting at $a$ to reach node $b$. Each case is illustrated by a corresponding figure.

\qquad \emph{\underline{Case I -- The walk reaches $v \in A$ before any node in $B \setminus A$:}} In this case, $h_{jAu}$ is equal to $\zeta_{jv} + \zeta_{vu}$, while $h_{jBu}$ is equal to the same quantity.  Thus $h_{jAu} = h_{jBu}$.  This case is illustrated in Figure \ref{fig:case_1}.

\begin{figure}[h]
\centering
\includegraphics[width=3in]{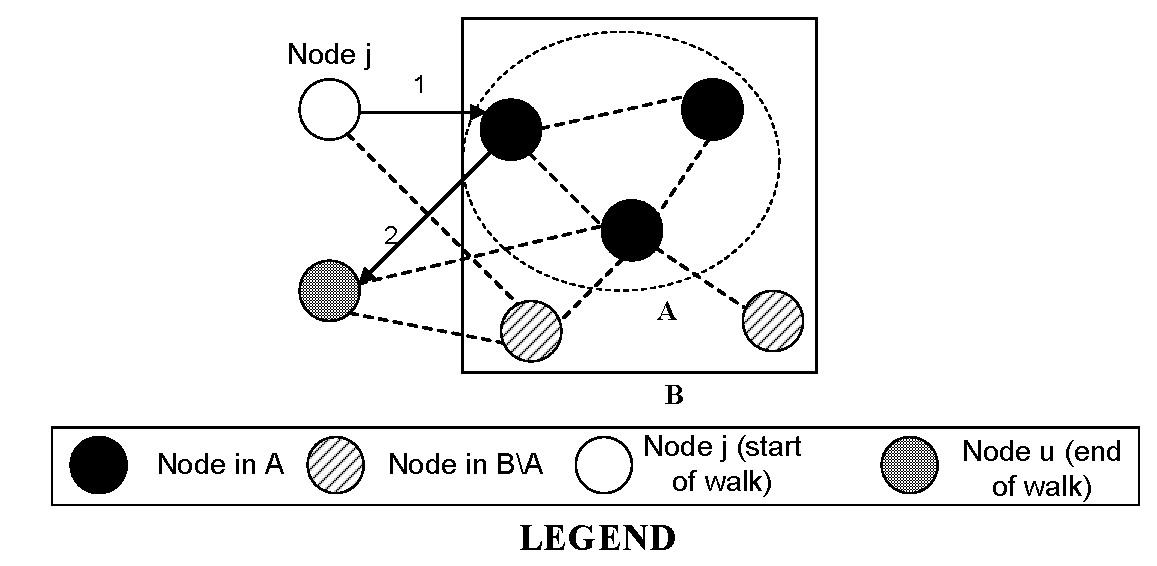}
\caption{Illustration of the proof that the commute time is supermodular, Case I, showing nodes $j$, $u$, and sets $A$ and $B$ with $A \subseteq B$.  In this case the walk, starting at node $j$, reaches set $A$ before any node in $B \setminus A$ and then continues to node $u$.}
\label{fig:case_1}
\end{figure}


\qquad \emph{\underline{Case II -- The walk reaches $t \in B \setminus A$, then $v \in A$ before $u$:}}  In this case, the $h_{jAu}$ is equal to $\zeta_{jv} + \zeta_{vu} = \zeta_{jt} + \zeta_{tv} + \zeta_{vu}$, since the walk reaches $t$ before $v$. Similarly, $h_{jBu}$ is equal to $\zeta_{jt} + \zeta_{tu} = \zeta_{jt} + \zeta_{tv} + \zeta_{vu}$, since the walk reaches $v$ after $t$ but before $u$.  Hence $h_{jAu} = h_{jBu}$ in Case II as well.  This case is illustrated in Figure \ref{fig:case_2}. 

\begin{figure}[h]
\centering
\includegraphics[width=3in]{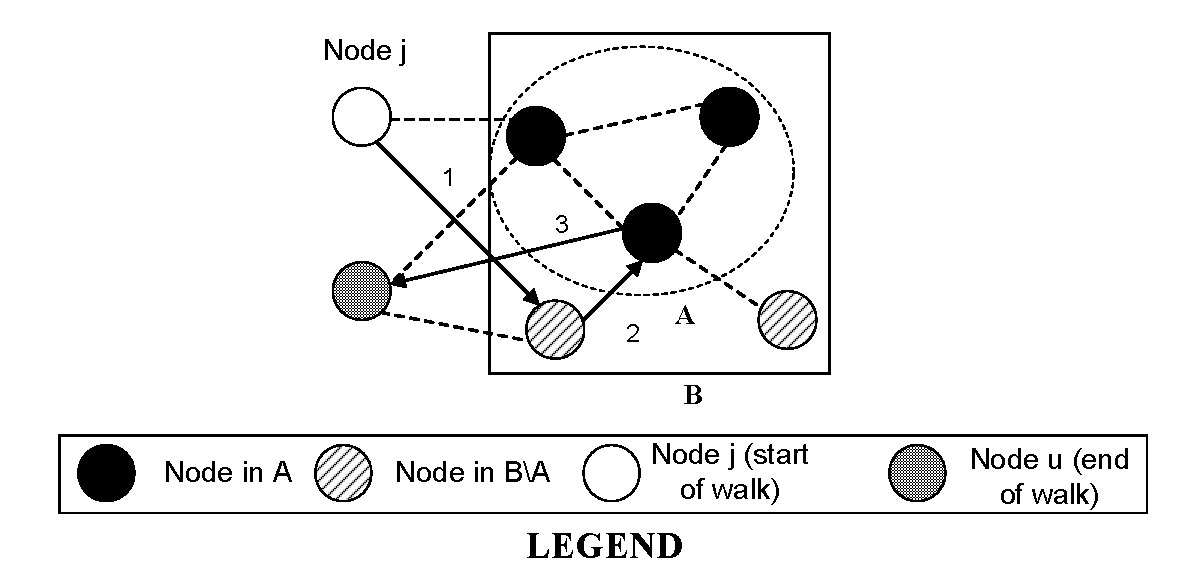}
\caption{Illustration of the proof that the commute time is supermodular, Case II, showing nodes $j$, $u$, and sets $A$ and $B$ with $A \subseteq B$.  In this case the walk, starting at node $j$, reaches set $B \setminus A$, then set $A$, and then node $u$.}
\label{fig:case_2}
\end{figure}

\qquad \emph{\underline{Case III -- The walk reaches $t \in B \setminus A$ and then $u$ before $v \in A$:}} $h_{jAu}$ is equal to $\zeta_{ju} + \zeta_{uv}$.  Since the walk reaches $t \in B \setminus A$ before $u$, this is equal to $\zeta_{jt} + \zeta_{tu} + \zeta_{uv}$.  However, since $h_{jBu}$ is the time for the walk to reach any node in $B$ and then travel to $u$, $h_{jBu} = \zeta_{jt} + \zeta_{tu}$.  Thus $h_{jAu} > h_{jBu}$ in Case III.  Case III is illustrated in Figure \ref{fig:case_3}.

\begin{figure}[h]
\centering
\includegraphics[width=3in]{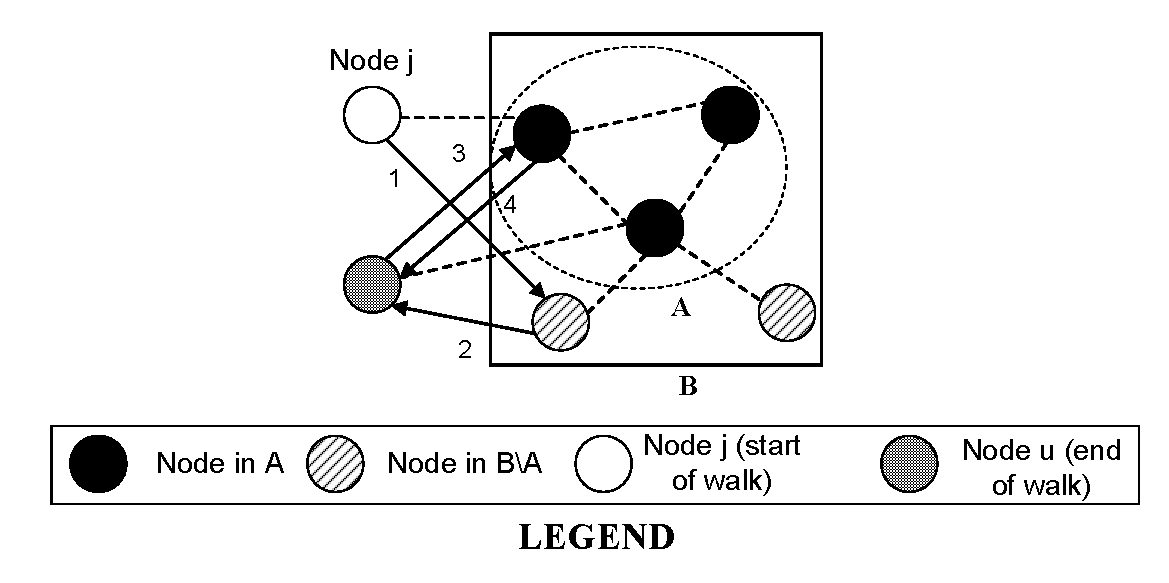}
\caption{Illustration of the proof that the commute time is supermodular, Case III, showing nodes $j$, $u$, and sets $A$ and $B$ with $A \subseteq B$.  In this case the walk, starting at node $j$, reaches set $B \setminus A$ and node $u$ before reaching set $A$.}
\label{fig:case_3}
\end{figure}
Together, these cases imply (\ref{eq:submitted_13}), and hence the supermodularity of $\kappa(S,u)$ as a function of $S$. 
\end{proof}

\end{document}